\documentclass[final,1p,times,authoryear]{elsarticle}

\makeatletter
\def\ps@pprintTitle{%
  \let\@oddhead\@empty
  \let\@evenhead\@empty
  \def\@oddfoot{}
  \let\@evenfoot\@oddfoot}
\makeatother

\usepackage{amsmath}
\usepackage{amsthm}
\usepackage{amssymb}
\usepackage{amsfonts}
\usepackage{color}
\usepackage{url}
\usepackage{hyperref}
\hypersetup{
    colorlinks=true,
}
\newtheorem{theorem}{Theorem}[section]
\newtheorem{lemma}[theorem]{Lemma}
\newtheorem{corollary}[theorem]{Corollary}

\newtheorem{proposition}[theorem]{Proposition}
\newtheorem{definition}[theorem]{Definition}

\newtheorem{example}[theorem]{Example}
\newtheorem{remark}[theorem]{Remark}

\usepackage{bm}
\usepackage{tikz}

\numberwithin{equation}{section}

\let\set\mathbb
\def\lc{\operatorname{lc}}
\def\cont{\operatorname{cont}}
\def\prim{\operatorname{prim}}
\def\supp{\operatorname{supp}}
\def\res{\operatorname{Res}}
\def\num{\operatorname{num}}

\def\bigO{\operatorname{O}}
\def\softO{\operatorname{O^{\sim}}}
\def\newt{\operatorname{Newt}}
\def\proj{\operatorname{Proj}}

\newcommand{\R}{{\mathsf R}}
\newcommand{\D}{{\mathsf D}}
\newcommand{\bx}[1]{{\bm x}^{\bm #1}}
\newcommand{\overbar}[1]{\mkern 1.5mu\overline{\mkern-1.5mu#1\mkern-1.5mu}\mkern 1.5mu}
\newcommand{\M}{{\mathsf M}}

\begin{document}

\begin{frontmatter}

\title{Efficient $q$-Integer Linear Decomposition\\ of Multivariate Polynomials}

\author{Mark Giesbrecht}
\address{Symbolic Computation Group, Cheriton School of Computer Science,
  University of Waterloo,\\
  Waterloo, ON, N2L 3G1, Canada}
\ead{mwg@uwaterloo.ca}

\author{Hui Huang}
\address{School of Mathematical Sciences, Dalian University of Technology,\\
  Dalian, Liaoning, 116024, China}
\ead{huanghui@dlut.edu.cn}

\author{George Labahn}
\address{Symbolic Computation Group, Cheriton School of Computer Science,
  University of Waterloo,\\
  Waterloo, ON, N2L 3G1, Canada}
\ead{glabahn@uwaterloo.ca}

\author{Eugene Zima}
\address{Physics and Computer Science, Wilfrid Laurier University,\\
  Waterloo, ON, N2L 3C5, Canada}
\ead{ezima@wlu.ca}

\begin{abstract}
  We present two new algorithms for the computation of the $q$-integer
  linear decomposition of a multivariate polynomial. Such a
  decomposition is essential for the treatment of $q$-hypergeometric
  symbolic summation via creative telescoping and for describing the
  $q$-counterpart of Ore-Sato theory.  Both of our algorithms require
  only basic integer and polynomial arithmetic and work for any unique
  factorization domain containing the ring of integers. Complete
  complexity analyses are conducted for both our algorithms and two
  previous algorithms in the case of multivariate integer polynomials,
  showing that our algorithms have better theoretical performances. A
  Maple implementation is also included which suggests that our
  algorithms are also much faster in practice than previous
  algorithms.
\end{abstract}

\end{frontmatter}

\section{Introduction}\label{SEC:intro}
Many objects in the ordinary shift world of symbolic summation find a
natural counterpart commonly called {\em $q$-analogues}. In a typical
situation, these are just slight adaptations of the original objects
but with involved variables promoted to exponents of an additional
parameter~$q$. Techniques for handling the originals often carry over
to their $q$-analogues with some subtle modifications.  One of the
reasons for interest in $q$-analogues is that, due to the extra
parameter~$q$, they have many counting interpretations which are
useful in combinatorics and analysis.  One is referred to the classic
books \citep{Andr1976,Andr1986} for the combinatorial and analytical
aspects of $q$-theory, as well as for some surprising applications
elsewhere in mathematics (see also \citep{BoYu2020}).

In this paper, we deal with the $q$-analogue of integer-linear
decompositions of polynomials and aim to provide an intensive
treatment for its computation in analogy to \citep{GHLZ2019}.
Surprisingly, although this $q$-analogue is obtained by modeling its
ordinary shift counterpart, the primary technique used in
\citep{GHLZ2019} can not be easily adapted to compute it due to
different structures. A new alternative technique will be presented in
this $q$-shift case.

In order to describe more details, we let $\D$ be a ring of
characteristic zero and let $\R=\D[q,q^{-1}]$ be its transcendental
ring extension by the indeterminate $q$. For $n$ discrete
indeterminates $k_1,\dots,k_n$ distinct from $q$, we know that
$q^{k_1},\dots,q^{k_n}$ are transcendental over $\R$.  We can then
consider polynomials in $q^{k_1},\dots,q^{k_n}$ over $\R$, all of
which form a well-defined ring denoted by $\R[q^{k_1},\dots,q^{k_n}]$.
We say an irreducible polynomial $p\in \R[q^{k_1},\dots,q^{k_n}]$ is
{\em $q$-integer linear} over $\R$ if there exists a univariate
polynomial $P\in\R[y]$ and two integer-linear polynomials
$\sum_{i=1}^n\alpha_ik_i, \sum_{i=1}^n\lambda_ik_i\in\set
Z[k_1,\dots,k_n]$ such that
\[
p(q^{k_1},\dots,q^{k_n}) = q^{\sum_{i=1}^n\alpha_ik_i}
P(q^{\sum_{i=1}^n\lambda_ik_i}).
\]
In order to avoid superscripts, we will write the indeterminates
$q^{k_1},\dots, q^{k_n}$ as the variables $x_1,\dots,x_n$ in the
sequel of the paper. Then the above definition can be rephrased as
follows. An irreducible polynomial $p\in\R[x_1,\dots,x_n]$ is called
{\em $q$-integer linear} over $\R$ if there exists a univariate
polynomial $P\in\R[y]$ and integers $\alpha_1,\dots,\alpha_n$,
$\lambda_1,\dots,\lambda_n$ such that
\begin{equation}\label{EQ:qild}
  p(x_1,\dots,x_n) = x_1^{\alpha_1}\cdots\, x_n^{\alpha_n}
  P(x_1^{\lambda_1}\cdots\, x_n^{\lambda_n}).
\end{equation}
Note that the indeterminate $q$ is hidden in the variables
$x_1,\dots,x_n$.  Since a common factor of the $\lambda_i$ can be
pulled out and absorbed into $P$, and a monomial can be merged into
$x_1^{\alpha_1}\cdots\, x_n^{\alpha_n}$ if necessary, we assume that
the integers $\lambda_1,\dots,\lambda_n$ have no common divisor, that
the last nonzero integer in the $\lambda_i$ is positive, that
$\lambda_i=0$ whenever $\deg_{x_i}(p)=0$ and that $P(0)\neq 0$. Such a
vector $(\lambda_1,\dots,\lambda_n)$, as well as such a polynomial
$P$, is unique. We call the vector $(\lambda_1,\dots,\lambda_n)$ the
{\em $q$-integer linear type} of $p$ and the polynomial $P$ its
{\em corresponding univariate polynomial}. Note that the resulting
$\alpha_1,\dots,\alpha_n$ all belong to $\set N$ since
$p\in\R[x_1,\dots,x_n]$ and $P\in\R[y]$. A polynomial in
$\R[x_1,\dots, x_n]$ is called {\em $q$-integer linear} (over $\R$) if
all its irreducible factors are $q$-integer linear, possibly with
different $q$-integer linear types. For a polynomial
$p\in\R[x_1,\dots, x_n]$, we can define its
{\em $q$-integer linear decomposition} by factoring into irreducible
$q$-integer linear or non-$q$-integer linear polynomials and
collecting irreducible factors having common types.

The class of $q$-integer linear polynomials plays a fundamental role
in the $q$-analysis of symbolic summation. For example, it is an
important ingredient of the $q$-analogue of the Ore-Sato theorem for
describing the structure of multivariate $q$-hypergeometric terms
\citep{DuLi2019}, which in turn serves as a promising indispensable
tool for settling a $q$-analogue of Wilf-Zeilberger's conjecture
\citep{WiZe1992a,ChKo2019}. Furthermore, the $q$-integer linearity of
polynomials is used to detect the applicability of the $q$-analogue of
Zeilberger's algorithm (also known as the method of creative
telescoping) for $q$-hypergeometric terms \citep{CHM2005}.

The full $q$-integer linear decomposition of polynomials is also very
useful.  On the one hand, it provides a natural way to determine
the $q$-integer linearity of a given polynomial. On the other hand, it
enables one to compute the $q$-analogue of Ore-Sato decomposition of a
given $q$-hypergeometic term, and can also be employed to develop a
fast creative telescoping algorithm for rational functions in the
$q$-shift setting in analogy to \citep{GHLZ2021}. Evidently, the
efficiency of the computation of $q$-integer linear decompositions
directly affects the utility of all these algorithms.

In contrast to the ordinary shift case \citep{AbLe2002,GHLZ2019,LiZh2013},
algorithms for computing the $q$-integer linear decomposition of a
multivariate polynomial are not very well developed. As far as we are
aware, there is only one algorithm available to compute such a
decomposition of a bivariate polynomial.  This algorithm was 
developed by \citet[\S 5]{Le2001} with an extended description
provided in \citep{LAG2001}. Except for using the same pattern as its
ordinary shift counterpart \citep{AbLe2002}, this algorithm takes use
of a completely different strategy, especially for finding $q$-integer
linear types. This is mainly because all $q$-integer linear types
appear as the exponent vectors of $p$, rather than as the coefficients
in the ordinary shift case.  The main idea used by \citet[\S 5]{Le2001}
is to first find candidates for $q$-integer linear types by computing
a resultant and then, for each candidate, extract the corresponding
univariate polynomial via bivariate GCD computations. Given the
algebraic machinery on which the algorithm is based, it is not clear
how one can directly generalize this to handle polynomials in more
than two variables.

The main contribution of this paper is a pair of new fast algorithms
for computing the $q$-integer linear decomposition of a multivariate
polynomial. Both algorithms will work for any unique factorization
domain containing all integers and for any polynomial with an
arbitrary number of variables.  The first approach follows the pattern
of the algorithm of Le but avoids the computation of resultants. More
precisely, this approach reduces the problem of finding candidates for
$q$-integer linear types to the well-studied geometric task of
constructing the Newton polytope of the given polynomial, implying
computations only using basic arithmetic operations ($+,-,\div,\times$)
of integers. It then computes each corresponding univariate polynomial
by a content computation.  As such we show that the $q$-analogue is
actually simpler than its ordinary shift counterpart in the sense
that, instead of finding rational roots of polynomials, one merely
needs to perform basic integer manipulations.

Our second approach uses a bivariate-based method. This scheme takes
the bivariate version of our previous algorithm, that is, the
algorithm for computing the $q$-integer linear decomposition of a
bivariate polynomial, as a base case and iteratively tackles only two
variables at a time until all variables are treated.  Clearly, our two
approaches coincide in the bivariate case.

An additional contribution is to use our bivariate-based scheme
(approach two) to extend the algorithm of Le so that it can readily
tackle polynomials in any number of variables. For the sake of
completeness, we also include another algorithm based on full
irreducible factorization.  This algorithm makes use of the
observation that the difference of exponent vectors of any two
monomials appearing in an irreducible $q$-integer linear polynomial,
say the polynomial $p$ of the form \eqref{EQ:qild}, must be a scalar
multiple of the $q$-integer linear type $(\lambda_1,\dots,\lambda_n)$.

In order to do a theoretical comparison we have analyzed the
worst-case running time complexity of our both approaches, as well as
that of the other two algorithms, in the case of polynomials over
$\set Z[q,q^{-1}]$. The analysis shows that the second approach is
superior to the first one when the given polynomial has more than two
variables. When restricted to the case of bivariate polynomials over
$\set Z[q,q^{-1}]$, the two approaches merge into one, which in turn
is considerably faster than the algorithm of Le and the algorithm
based on factorization. In addition, we also give experimental results
which verify our complexity comparisons.

The remainder of the paper proceeds as follows. Background and basic
notions required in the paper are provided in the next section.  Our
two new approaches for computing $q$-integer linear decompositions of
multivariate polynomials are given successively in
Sections~\ref{SEC:1stapproach} and \ref{SEC:2ndapproach}. The
following section provides a complexity comparison of our two
algorithms, the algorithm of Le and the factorization-based algorithm.
The paper ends with an experimental comparison among all algorithms,
along with a conclusion section.

\section{Preliminaries: polynomials and Newton polytopes}
\label{SEC:prelim}
Throughout the paper, we let $\D$ be a unique factorization domain
(UFD) of characteristic zero with $\R = \D[q,q^{-1}]$ denoting the
transcendental ring extension by an indeterminate~$q$. Note that a
domain of characteristic zero always contains the ring of integers
$\set Z$ as a subdomain. Let $\R[x_1,\dots,x_n]$ be the ring of
polynomials in $x_1,\dots,x_n$ over $\R$, where $x_1,\dots,x_n$ are
variables distinct from~$q$.  We reserve the variables $x$ and $y$ as
synonyms for $x_1$ and $x_2$, respectively, so as to avoid subscripts
in the case when $n\leq 2$.

Let $p$ be a polynomial in~$\R[x_1,\dots,x_n]$. Throughout this paper
we will order monomials in $\R[x_1,\dots,x_n]$ using a pure
lexicographic order in $x_1\prec\dots\prec x_n$. For this order we let
$\lc(p)$ and $\deg(p)$ denote the leading coefficient and the total
degree, respectively, of $p$ with respect to~$x_1,\dots,x_n$. We
follow the convention that $\deg(0)=-\infty$. We say that $p$ is
{\em monic} (over $\R$) if $\lc(p)=1$. The {\em content} of $p$ (over
$\R$), denoted by $\cont(p)$, is the greatest common divisor (GCD)
over $\R$ of the coefficients of $p$ with respect to $x_1,\dots,x_n$
with $p$ being {\em primitive} if $\cont(p)=1$. The {\em primitive part}
$\prim(p)$ of $p$ (over $\R$) is defined as $p/\cont(p)$. For brevity,
we will omit the domain if it is clear from the context. In certain
instances, we also need to consider the above notions with respect to
a subset of the $n$ variables. In these cases, we will either specify
the relevant domain or indicate the related variables as subscripts of
the corresponding notion. For example, $\lc_{x_1,x_2}(p)$,
$\deg_{x_1,x_2}(p)$, $\cont_{x_1,x_2}(p)$ and $\prim_{x_1,x_2}(p)$
denote each function but applied to a polynomial $p$ viewing it as a
polynomial in $x_1,x_2$ over the domain $\R[x_3,\dots,x_n]$.

In order to obtain a canonical representation, we introduce the notion
of $q$-primitive polynomials in the univariate case. A polynomial
$p\in \R[y]$ is called {\em $q$-primitive} if it is primitive over
$\R$ and its constant term $p(0)$ is nonzero.  Note that this concept
is a ring counterpart of $q$-monic polynomials introduced by
\cite{PaRi1997}.  Clearly, any factor of a $q$-primitive polynomial in
$\R[y]$ is again $q$-primitive.

The Newton polytope of multivariate polynomials plays a crucial role 
in our algorithms. In what follows, we recall some terminology and 
results on convex polytopes from a polynomial point of view. For a
more general theory, one is referred to, for example, \citep{Grun2003}.

In order to simplify notations, we employ bold letters, say $\bm i$,
for a column vector $(i_1,\dots,i_n)^T$ in the Euclidean space $\set R^n$,
and the {\em multi-index convention} $\bx{i}$ for the monomial
$x_1^{i_1}\cdots\, x_n^{i_n}$ if $\bm i \in \set Z^n$. The zero vector
in $\set R^n$ is denoted by boldface~$\bm 0$. Taking advantage of this
boldface notation, we later write $\R[\bm x]$ and $\R[\bm x,\bm x^{-1}]$
for the polynomial ring $\R[x_1,\dots,x_n]$ and the Laurent polynomial
ring $\R[x_1,x_1^{-1},\dots,x_n,x_n^{-1}]$, respectively.

Let $p\in \R[\bm x]$ be a polynomial of the form $\sum_{\bm i}a_{\bm i}
\bx{i}$ with $a_{\bm i} \in \R$, having finitely many nonzero terms.
The {\em support} of $p$, denoted by $\supp(p)$, is defined as the set
of indices $\bm i\in \set N^n$ with the property that the corresponding
coefficient $a_{\bm i}$ is nonzero. Clearly, $\supp(p)$ is a finite
set in $\set N^n$, and it is empty if and only if $p=0$. An exponent
vector $\bm i$ of $p$ can be considered as a point in~$\set R^n$. The
convex hull of the set $\supp(p)$ in $\set R^n$ is then known as the
{\em Newton polytope} of $p$, denoted by~$\newt(p)$. By convention,
$\newt(0)$ is the empty set.

For two sets $A$ and $B$ in $\set R^n$, their {\em Minkowski sum} is 
defined as the set 
\[
A + B = \{\bm a + \bm b\mid \bm a \in A, \bm b\in B\}.
\]
The following well-known result, due to \cite{Ostr1921,Ostr1975},
reveals the relation between the Newton polytope of a polynomial and
those of its factors.
\begin{lemma}[{\citep{Ostr1921,Ostr1975}}]\label{LEM:Msum}
  Let $f,g \in \R[\bm x]$. Then $\newt(fg) = \newt(f) + \newt(g)$.
\end{lemma}

It proves convenient to extend the notion of Newton polytopes to
Laurent polynomials in the ring $\R[\bm x,\bm x^{-1}]$. Notice that
any Laurent polynomial from $\R[\bm x,\bm x^{-1}]$ can be written as
the form $\bx{\alpha} p$ for some $\bm \alpha\in \set Z^n$ and $p\in
\R[\bm x]$. Thus the {\em Newton polytope} of the given Laurent
polynomial is defined to be the translation $\newt(p) + \bm \alpha$ of
$\newt(p)$ by $\bm \alpha$. Evidently, Lemma~\ref{LEM:Msum} literally
carries over to Laurent polynomials.
\begin{lemma}\label{LEM:LMsum}
  Let $f,g \in \R[\bm x,\bm x^{-1}]$. Then $\newt(fg) = \newt(f) +
  \newt(g)$.
\end{lemma}
We will consider faces of Newton polytopes. Let $C$ be a Newton
polytope of a certain Laurent polynomial over $\R$. A hyperplane
$H = \{\bm x\in \set R^n\mid \bm a^T \bm x = b\}$ with $\bm a\in
\set R^n \setminus\{\bm 0\}$ and $b\in \set R$ is called a
{\em supporting hyperplane} of $C$ with {\em outward normal} $\bm a$
if $H\cap C \neq \emptyset$ and $\bm a^T \bm x \leq b$ for all $\bm x
\in C$. We call the intersection $H\cap C$ a {\em face} of~$C$. By 
convention, $\emptyset$ and $C$ are called {\em improper faces} of~$C$.
The faces of dimension zero and one are also called {\em vertices} and
{\em edges}, respectively. Note that for any nonzero vector $\bm a \in
\set R^n$, there exists a unique supporting hyperplane of $C$ with
outward normal $\bm a$ (cf.\ \cite[Theorem~8, Page 15]{Grun2003}). We
then refer to the intersection of this supporting hyperplane and $C$
as the face of $C$ determined by the outward normal $\bm a$.
\begin{lemma}\label{LEM:facesum}
  Let $f,g \in R[\bm x,\bm x^{-1}]$ and $\bm a\in\set R^n\setminus\{0\}$.
  Then $F_{fg,\bm a} = F_{f,\bm a} + F_{g,\bm a}$, where $F_{f,\bm a}$ is
  the face of $\newt(f)$ determined by the outward normal $\bm a$.
\end{lemma}
\begin{proof}
  By Lemma~\ref{LEM:LMsum}, $\newt(f) = \newt(f) + \newt(g)$. The
  assertion is then a direct result of \cite[Theorem~1, Page 317]{Grun2003}.
\end{proof}

\section{$q$-Integer linear decomposition: the first approach}
\label{SEC:1stapproach}

We are interested in finding the following decomposition of a
polynomial, something briefly alluded to in the introduction.
\begin{definition}\label{DEF:qild}
  Let $p\in \R[\bm x]$ be a polynomial admitting the decomposition
  \begin{equation}\label{EQ:multiqild}
    p = c\, \bx{\alpha} P_0 \prod_{i=1}^m P_i(\bx{\lambda_i}),
  \end{equation}
  where $c\in \R$, $m\in \set N$, $\bm{\alpha}\in \set N^n$,
  $\bm{\lambda}_i\in \set Z^n\setminus\{\bm 0\}$, $P_0\in\R[\bm x]$
  and $P_i\in \R[y]$. Then \eqref{EQ:multiqild} is called the
  {\em $q$-integer linear decomposition} of $p$ (over $\R$) if
  \begin{itemize}
  \item[(1)] $P_0$ is primitive and none of its irreducible factors of
    positive total degree is $q$-integer linear;
  \item[(2)] each $P_i$ is $q$-primitive and of positive degree;
  \item[(3)] each $\bm{\lambda}_i$ satisfies the conditions that
    $\gcd(\lambda_{i1},\dots,\lambda_{in})=1$ and its rightmost
    nonzero coordinate is positive
    \footnote{As mentioned in the introduction, the positivity of the
    rightmost nonzero coordinate of $\bm \lambda_i$ required here can
    be easily obtained and is used to make such a vector unique.};
  \item[(4)] the $\bm{\lambda}_i$ are pairwise distinct.
  \end{itemize}
  We call each $\bm{\lambda}_i$ a {\em $q$-integer linear type} of $p$
  and $P_i$ its {\em corresponding univariate polynomial}.
\end{definition}

Evidently, $p$ is $q$-integer linear if and only if $P_0$ is a unit of
$\R$ in \eqref{EQ:multiqild}. By full factorization, we see that every
polynomial admits a $q$-integer linear decomposition.  Moreover, this
decomposition is unique up to the order of factors and multiplication
by units of $\R$, according to the uniqueness of full factorization
and that of the $q$-integer linear type of an irreducible polynomial.

Let $p\in \R[\bm x]$ be a polynomial of positive total degree. Without
loss of generality, we assume that $p$ is primitive with respect to
any variable from $\{x_1,\dots,x_n\}$. Otherwise, we may replace $p$
by the remaining part after iteratively removing from $p$ its content
with respect to $x_i$ for all $i = 1,\dots,n$. Note that all these
removed contents are polynomials over $\R$ having at most $(n-1)$
variables and hence can be dealt with recursively, knowing that
univariate polynomials are all $q$-integer linear. With this set-up,
$p$ admits the $q$-integer linear decomposition of the form
\eqref{EQ:multiqild}, in which $c=1$, $\alpha_n=0$ and none of the
types $\bm\lambda_i$ has zero coordinates. In order to compute such a
decomposition, we mimic the strategy of \cite{AbLe2002} in the
ordinary shift case, that is, we first find all possible candidates
for $q$-integer linear types and then extract the corresponding
univariate polynomial for each type.

\subsection{Candidates for $q$-integer linear types}
Observe that all $q$-integer linear types $\bm \lambda_i$ in
\eqref{EQ:multiqild} appear as exponent vectors, and the Newton
polytope of each $P_i(\bx{\lambda_i})$ is just a line segment. This
leads us to investigate edges of the Newton polytope of the given
polynomial.

For this purpose, we assign a direction to each line segment in~$\set
R^n$. Let $\bm u,\bm v\in\set R^n$ with $\bm u \neq \bm v$ and let
$[\bm u,\bm v] = \{t\bm u + (1-t)\bm v\mid t\in\set R, 0\leq t\leq
1\}$ denote the line segment connecting~$\bm u,\bm v$. A nonzero
vector $\bm\lambda \in \set R^n$ is called the {\em direction vector}
of $[\bm u,\bm v]$ if $\bm u - \bm v = t\bm\lambda$ for some $t\in\set
R$, $\gcd(\lambda_1,\dots,\lambda_n) = 1$ and the rightmost nonzero
coordinate of $\bm \lambda$ is positive. As before, the requirement on
the positivity of the last nonzero coordinate guarantees
the uniqueness of such a direction vector. Clearly, two parallel
(nondegenerate) line segments share the same direction vector, and
vice versa.
\begin{lemma}\label{LEM:newt}
  Let $p\in \R[\bm x]\setminus\R$ with $\cont_{x_1}(p)=\cdots=
  \cont_{x_n}(p)=1$, and assume that it admits the $q$-integer linear
  decomposition \eqref{EQ:multiqild}. Then for any $i\in \set N$ with
  $1\leq i\leq m$, the Newton polytope of $p$ possesses an edge of the
  direction vector $\bm \lambda_i$. Moreover, if $\newt(p)$ is not a
  line segment then there are at least two such edges.
\end{lemma}
\begin{proof}
  There is nothing to show when $m = 0$, so assume that $m > 0$. We
  merely show the assertions for $i = m$, and then the lemma follows 
  by symmetry.

  Let $p^*=\bx{\alpha} P_0\prod_{i=1}^{m-1}P_i(\bx{\lambda_i})$. Then
  $p^*\in \R[\bm x]\setminus\{0\}$, and by \eqref{EQ:multiqild},
  \begin{equation}\label{EQ:pform}
    p=p^* P_m(\bx{\lambda_m}).
  \end{equation}
  Notice that $\newt(P_m(\bx{\lambda_m}))$ is a line segment in $\set
  R^n$ with direction vector $\bm\lambda_m$. Then for any nonzero
  vector $\bm a\in\set R^n$ with $\bm a^T \bm \lambda_m = 0$, the
  supporting hyperplane of $\newt(P_m(\bx{\lambda_m}))$ determined by
  the outward normal $\bm a$ contains the whole polytope. This means
  that $\newt(P_m(\bx{\lambda_m}))$ itself is the (improper) edge
  determined by such an outward normal.

  In order to show the first assertion, it then amounts to finding a
  nonzero vector $\bm a\in\set R^n$ with $\bm a^T \bm \lambda_m = 0$
  such that the face of $\newt(p^*)$ determined by the outward normal
  $\bm a$ is either a vertex or an edge parallel to
  $\newt(P_m(\bx{\lambda_m}))$.  The rest then follows by
  \eqref{EQ:pform}, Lemma~\ref{LEM:facesum} and the observation that
  the Minkowski sum of a line with a point or another parallel line is
  again a line parallel to the original line.

  By an affine coordinate transformation if necessary, we may assume
  without loss of generality that $\bm \lambda_m$ is equal to the
  $n$-th unit vector $\bm e_n = (0,\dots,0,1)^T\in\set R^n$. Then
  $\newt(P_m(\bx{\lambda_m}))$ is contained by the $x_n$-axis.  We now
  consider the projection of $\newt(p^*)$ onto the hyperplane $\{\bm
  x\in\set R^n\mid x_n = 0\}$ in the direction of $\bm\lambda_m=\bm e_n$,
  that is,
  \[
  \proj_n(p^*) = \{\bm x \in \set R^n\mid x_n = 0 \ \text{and}\
  \bm x + t\bm e_n \in \newt(p^*)\ \text{for some}\ t\in\set R\}.
  \]
  This is again a Newton polytope by \cite[Theorem~8, Page 74]{Grun2003}.
  Since $p^*$ is nonzero, $\newt(p^*)$ is nonempty, and so is $\proj_n(p^*)$.
  Let $\tilde{\bm v}$ be a vertex of $\proj_n(p^*)$. Then by definition,
  there exists a hyperplane $H$ of the form
  $H = \{\bm x\in\set R^n\mid \bm a^T\bm x = b\}$ for
  $\bm a\in\set R^n\setminus\{\bm 0\}$ with $a_n = 0$ and $b\in\set R$
  such that $H\cap\proj_n(p^*) = \{\tilde{\bm v}\}$ and $\bm a^T\bm x\leq b$
  for all $\bm x\in \proj_n(p^*)$. Since $\tilde{\bm v}\in\proj_n(p^*)$,
  there exists a number $t\in\set R$ such that
  $\tilde{\bm v}+t\bm e_n\in\newt(p^*)$. Among these numbers, let
  $t_1,t_2\in\set R$ be the minimum and maximum ones, respectively.
  Note that $t_1,t_2$ are not necessarily distinct. Let $\bm u =
  \tilde{\bm v}+t_1\bm e_n$ and $\bm v = \tilde{\bm v}+t_2\bm e_n$.
  Then the line segment $[\bm u,\bm v]$, possibly being a point when
  $t_1 = t_2$, is parallel to the $x_n$-axis and contained in $\newt(p^*)$
  by convexity.
  
  Evidently, $\bm a^T\bm \lambda_m = \bm a^T \bm e_n = 0$. We claim
  that $[\bm u,\bm v]$ is the face of $\newt(p^*)$ determined by the
  outward normal $\bm a$, which will complete the proof of the first
  assertion. In other words, we aim to prove that
  \[
  H\cap\newt(p^*) = [\bm u, \bm v] \quad\text{and}\quad
  \bm a^T\bm x \leq b \ \text{for all}\ \bm x\in\newt(p^*).
  \]
  Let $\bm x\in\newt(p^*)$ and $\tilde{\bm x} = (x_1,\dots,x_{n-1},0)$.
  Then $\bm a^T \bm x = \bm a^T\tilde{\bm x}\leq b$ as $a_n = 0$ and
  $\tilde{\bm x}\in\proj_n(p^*)$. To see the inclusion $H\cap\newt(p^*)
  \subset [\bm u, \bm v]$, we further assume that $\bm x \in H\cap\newt(p^*)$. 
  Thus $\tilde{\bm x} \in H\cap\proj_n(p^*) = \{\tilde{\bm v}\}$. This means
  that $\tilde{\bm x} = \tilde{\bm v}$. By the minimality of $t_1$ and 
  maximality of $t_2$, we know that $\bm x\in [\bm u, \bm v]$. 
  The opposite direction $H\cap\newt(p^*)\supset [\bm u, \bm v]$ is 
  clear from definition.
  
  Moreover, assume that $\newt(p)$ is not a line segment. Then
  $\newt(p^*)$ cannot be a point or a line segment parallel to
  $\newt(P_m(\bx{\lambda_m}))$ by \eqref{EQ:pform} and
  Lemma~\ref{LEM:LMsum}. This implies that $\proj_n(p^*)$ has at
  least two different vertices. Taking another vertex of
  $\proj_n(p^*)$ distinct from $\tilde{\bm v}$ and arguing along
  similar lines as above yields another edge of $\newt(p)$ which has
  the direction vector $\bm \lambda_m$. The lemma therefore follows.
\end{proof}

From the above lemma, one sees that the direction vectors of edges of
$\newt(p)$ exhaust all possible choices of $q$-integer linear types.
When $\newt(p)$ is not a line segment, one can restrict attention to
those vectors with multiple occurrences. Note that in our application, 
the Newton polytope of a given polynomial will be described by the set 
of its edges. Such a set can be easily deduced from the
face lattice or the vertex-facet incidence matrix of the given Newton
polytope, for which algorithms from computational geometry are well
developed; see \cite[Chapter~26]{GOT2018} and the references therein.

Given a set of points with cardinality $s\in\set N$, it is known that
the number of edges of the convex hull of this set is bounded
by~${s\choose 2}$ (cf.\ \citep[Theorem~2, Page 194]{Grun2003}). Thus
Lemma~\ref{LEM:newt} might offer us a superset of $q$-integer linear
types of cardinality $\bigO(s^2)$ in the worst case. The following
lemma, however, helps us bring it down to $\bigO(s)$.
\begin{lemma}\label{LEM:multiqiltype}
  With the assumptions of Lemma~\ref{LEM:newt}, for any $i\in\set N$
  with $1\leq i\leq m$ and for any $\bm j\in \supp(p)$, there exists
  another vector $\tilde{\bm j} \in\supp(p)$ such that the line
  segment $[\bm j, \tilde{\bm j}]$ has the direction vector
  $\bm\lambda_i$, or equivalently, $\bm j-\tilde{\bm j} = k\bm
  \lambda_i$ for some nonzero integer $k$.
\end{lemma}
\begin{proof}
  There is nothing to show when $m = 0$, so assume that $m > 0$. By
  symmetry, it suffices to show that the assertion holds for $i = m$.
  
  Again, we take $p^*=\bx{\alpha} P_0\prod_{i=1}^{m-1}P_i(\bx{\lambda_i})$
  and derive the decomposition \eqref{EQ:pform} of $p$. Notice that 
  $\supp(p)$ is nonempty as $p\neq 0$. Let $\bm j\in\supp(p)$.
  It follows from \eqref{EQ:pform} that there is $\bm j^*\in\supp(p^*)$ 
  and $k^*\in \supp(P_m)$ such that $\bm j=\bm j^*+k^*\bm\lambda_m$. 
  Now consider the set
  \[
  S = \{\bar{\bm j}\in\supp(p^*)\mid \bar{\bm j}
  =\bm j^*+k\bm\lambda_m \ \text{for some}\ k\in\set Z\}.
  \]
  Then there exist $p_1^*,p_2^*\in\R[\bm x]$ with $\supp(p_1^*)=S$ and
  $\supp(p_2^*)=\supp(p^*)\setminus S$ such that $p^*=p_1^*+p_2^*$.
  It is evident that $\bm j^*\in S$. Thus $S$ is nonempty and then
  $p_1^*$ is nonzero. Let $\bm \alpha^*\in S$ be such that any element
  of $S$ can be written as $\bm \alpha^*+k\bm\lambda_m$ for some $k\in
  \set N$, or equivalently, any monomial present in $p_1^*$ takes the
  form $\bx{\alpha^*+k\bm \lambda_m}$ for some $k\in \set N$. It then
  follows that there exists a nonzero univariate polynomial
  $P^*\in\R[y]$ such that $p_1^*=\bx{\alpha^*}P^*(\bx{\lambda_m})$.

  On the other hand, by noticing that for any $\bar{\bm j}\in\supp(p_2^*)
  =\supp(p^*)\setminus S$, we have $\bar{\bm j}\neq \bm j^*+k\bm\lambda_m$ 
  for all $k\in\set Z$. Hence, $p$ can be decomposed as $p=f+g$, where
  $f=p_1^*P_m(\bx{\lambda_m})$ and $g=p_2^*P_m(\bx{\lambda_m})$ with
  $\supp(f)\cap\supp(g)=\emptyset$. As a consequence,
  $\supp(p)=\supp(f)\uplus\supp(g)$. Since $\bm j=\bm
  j^*+k^*\bm\lambda_m$, we have $\bm j\in \supp(f)$. Notice that
  $p_1^*=\bx{\alpha^*}P^*(\bx{\lambda_m})$.  So $f=\bx{\alpha^*}\tilde
  P(\bx{\lambda_m})$ with $\tilde P=P^*P_m \in\R[y]\setminus\{0\}$.
  Then there exists $k\in\supp(\tilde P)$ such that $\bm j=\bm
  \alpha^*+k\bm\lambda_m$.  Since $P_m$ is $q$-primitive and of
  positive total degree, it possesses more than one monomial, and
  hence so does $\tilde P$. This implies that there is another element
  $\tilde k\in \supp(\tilde P)$ distinct from $k$. Let $\tilde{\bm j}
  =\bm \alpha^*+\tilde k\bm\lambda_m$.  Then $\tilde{\bm j}\in\supp(f)
  \subset\supp(p)$ and $\bm j-\tilde{\bm j} = (k-\tilde k)\bm\lambda_m$.
  This concludes the proof.
\end{proof}
Combining Lemmas~\ref{LEM:newt} and \ref{LEM:multiqiltype} suggests a
simple geometric way to find candidates for all $q$-integer linear
types of a given polynomial.
\begin{proposition}\label{PROP:multiqiltypes}
  With the assumptions of Lemma~\ref{LEM:newt}, let $\Lambda_1$ be the
  multiset of direction vectors of edges of $\newt(p)$ having no zero
  coordinates. Let $\bm v\in\supp(p)$ be fixed and let $\Lambda_2$ be
  the set consisting of direction vectors of line segments connecting
  $\bm v$ and all other points in $\supp(p)$ which have no zero
  coordinates.
  \begin{itemize}
  \item[(1)] If the cardinality of $\Lambda_1$ is one then $p$ is
    $q$-integer linear of type $\bm\lambda\in \Lambda_1$.
  \item[(2)] Otherwise, let $\Lambda_1^*$ be the subset of $\Lambda_1$
    composed of elements with multiple occurrences. Then the
    intersection $\Lambda_1^*\cap \Lambda_2$ constitutes a superset of
    $q$-integer linear types of~$p$. Moreover, with $s\in \set N$
    denoting the cardinality of $\supp(p)$, this superset has no more
    than $s-1$ elements in total.
  \end{itemize}
\end{proposition}
Let $p$ be as given in Lemma~\ref{LEM:newt} and assume further that
$p$ is $q$-integer linear. Then one sees from the decomposition
\eqref{EQ:multiqild} and Lemma~\ref{LEM:LMsum} that $\newt(p)$ is the
Minkowski sum of finitely many line segments. Such a polytope is
called a {\em zonotope} in the literature. Zonotopes form an
especially interesting and important class of convex polytopes; we
refer to \cite[Lecture~7]{Zieg1995} for more information. One of the
key features of the zonotope $\newt(p)$ is that the direction vectors
of its edges are exactly those of its zones (namely the line segments
present in the Minkowski sum), which, in our context, are all
$q$-integer linear types $\bm \lambda_1,\dots, \bm \lambda_m$ from
\eqref{EQ:multiqild}. We therefore obtain the following necessary
condition for a polynomial to be $q$-integer linear.
\begin{proposition}\label{PROP:multinecessity}
  With the assumptions of Lemma~\ref{LEM:newt}, further assume that
  $p$ is $q$-integer linear. Then $\newt(p)$ is a zonotope and none of
  the direction vectors of edges of $\newt(p)$ has zero coordinates.
  As a consequence, for any integer $i$ with $1\leq i\leq n$, there
  exists a unique vector in $\supp(p)$ whose $i$-th coordinate takes
  extremum value.
\end{proposition}
\begin{proof}
  Notice that none of the $q$-integer linear types of $p$ has zero
  coordinates. The first assertion is thus a direct result of the
  discussion preceding the proposition. In terms of the second
  assertion, we only show the argument on minimality for $i = n$, that
  is, we will prove that there exists only one vector in $\supp(p)$
  whose $i$-th coordinate attains minimum. The rest follows by symmetry.
  
  We proceed with using proof by contradiction. Suppose that there are
  at least two vectors in $\supp(p)$ whose $n$-th coordinate is equal
  to $\min_{\bm x\in\supp(p)}\{x_n\}$. Let $\bm a\in\supp(p)$ be one
  of these vectors. We claim that $H := \{\bm x\in\set R^n\mid
  -x_n=-a_n\}$ is a supporting hyperplane of $\newt(p)$. By the
  minimality of $a_n$, we know that $-x_n\leq -a_n$ for all $\bm
  x\in\supp(p)$. It then follows from the convexity of $\newt(p)$ that
  $-x_n\leq -a_n$ for all $\bm x\in\newt(p)$. Since $\bm a \in H\cap
  \newt(p)\neq\emptyset$, the claim holds.
  
  Let $F = H\cap \newt(p)$. Then $F$ is a face of $\newt(p)$ by the
  claim and thus is itself a Newton polytope by
  \cite[Proposition~2.3(i)]{Zieg1995}. By assumption, $F$ has at least
  two points and then possesses an edge, say $[\bm u, \bm v]$ for $\bm
  u,\bm v\in\supp(p)$. By \cite[Proposition~2.3~(iii)]{Zieg1995},
  $[\bm u,\bm v]$ is also an edge of $\newt(p)$, whose direction
  vector has zero $n$-th coordinate since $\bm u,\bm v\in F\subset H$,
  a contradiction with the first assertion.
\end{proof}

\subsection{Computation of univariate polynomials}
With candidates for the $q$-integer linear types at hand, we are 
able to find the corresponding univariate polynomials based on a
$q$-counterpart of \cite[Proposition~3.2]{GHLZ2019}.
\begin{proposition}\label{PROP:multitype2poly}
  With the assumptions of Lemma~\ref{LEM:newt}, let
  $\bm{\lambda}\in\set Z^n$ with $\gcd(\lambda_1,\dots,\lambda_n)=1$,
  $\lambda_1,\dots,\lambda_{n-1}$ not all zero and $\lambda_n>0$.  Let
  $P^*\in\R[y]$ be the content with respect to $x_1,\dots,x_{n-1}$ of
  the numerator of $p(x_1^{\lambda_n},\dots,x_{n-1}^{\lambda_n},
  yx_1^{-\lambda_1}\cdots\,x_{n-1}^{-\lambda_{n-1}})$. If $P^* \notin\R$
  then $\bm{\lambda}$ is a $q$-integer linear type of $p$ with
  corresponding univariate polynomial $P^*(y^{1/\lambda_n})\in\R[y]$.
  Otherwise, $\bm{\lambda}$ is not a $q$-integer linear type of $p$.
\end{proposition}
In order to prove the above proposition, we first need to introduce
some basic notions and lemmas.  In the sequel of this subsection, we
let $\set K$ denote the quotient field of $\R$ and consider
polynomials in $x_n$ over the field $\set K(x_1,\dots,x_{n-1})$, all
of which form the ring $\set K(x_1,\dots,x_{n-1})[x_n]$. It is
convenient to extend the definition of content and primitive part to
polynomials in this setting. Let $p\in \set K(x_1,\dots,x_{n-1})[x_n]$
be of the form $\sum_{i=0}^d(a_i/b)x_n^i $ for $d\in \set N$ and
$a_i,b\in\R[x_1,\dots,x_{n-1}]$. Then the {\em content}
$\cont_{x_n}(p)$ of $p$ with respect to $x_n$ is defined as
$\gcd(a_0,\dots,a_d)/b$ and the corresponding {\em primitive part}
$\prim_{x_n}(p)=p/\cont_{x_n}(p)$.  Evidently,
$\prim_{x_n}(p)\in\R[\bm x]$.  The definition of leading coefficient
and degree extends to polynomials in $\set K[x_1,\dots,x_n]$ in a
natural manner.

\begin{lemma}\label{LEM:sub}
  Let $P\in\R[y]\setminus\R$ with $P(0)\neq 0$ and let
  $\bm{\lambda}\in\set Z^n$ with $\gcd(\lambda_1,\dots,\lambda_n)=1$,
  $\lambda_1,\dots,\lambda_{n-1}$ not all zero and $\lambda_n>0$.
  Then
  \begin{itemize}
  \item[(i)] for any factor $f\in\set K(x_1,\dots,x_{n-1})[x_n]$ of
    $P(\bx{\lambda})$ which is monic and irreducible over
    $\set K(x_1,\dots,x_{n-1})$, there exists $c\in\set K$,
    $\alpha_1,\dots,\alpha_{n-1}\in\set Z$ and a factor $g\in\R[y]$ of
    $P$ such that
    $f=cx_1^{\alpha_1}\cdots\,x_{n-1}^{\alpha_{n-1}}g(\bx{\lambda})$.
    Moreover, $0<\deg(g)=\deg_{x_n}(f)/\lambda_n$.
  \item[(ii)] $P$ is irreducible over $\R$ if and only if
    $P(\bx{\lambda})$ is irreducible over $\set K(x_1,\dots,x_{n-1})$
    if and only if $\prim_{x_n}(P(\bx{\lambda}))$ is irreducible
    over~$\R$.
  \end{itemize}
\end{lemma}
\begin{proof}
  (i) Since $P(0)\neq 0$, all its roots in the algebraic closure
  $\overbar{\set K}$ of the field $\set K$ are nonzero. In order to
  prove the assertion, it is sufficient to show that for any root
  $r\in\overbar{\set K}$ of $P$, the polynomial $\bx{\lambda}-r$ is
  irreducible over $\overbar{\set K}(x_1,\dots,x_{n-1})$.  For then,
  since $f\in\set K(x_1,\dots,x_{n-1})[x_n]$ is a monic and
  irreducible factor of $P(\bx{\lambda})$, it factors completely into
  irreducibles in $\overbar{\set K} (x_1,\dots,x_{n-1})[x_n]$ as
  follows
  \[
  f=\prod_{i=1}^s(x_1^{-\lambda_1}\cdots\,x_{n-1}^{-\lambda_{n-1}})
  (\bx{\lambda}-r_i)
  =(x_1^{-\lambda_1}\cdots\,x_{n-1}^{-\lambda_{n-1}})^s
  \prod_{i=1}^s(\bx{\lambda}-r_i),
  \]
  where $s\in \set N$ with $s\leq \deg(P)$ and the
  $r_i\in\overbar{\set K}$ are roots of $P$, and thus the assertion
  directly follows by letting $g(y)=\prim_y(\prod_{i=1}^s(y-r_i))$.
  
  Let $r\in\overbar{\set K}$ be a root of $P$ and suppose that
  $\bx{\lambda}-r$ is reducible over $\set K(x_1,\dots,x_{n-1})$.
  Then we have $\lambda_n>1$. Consider the algebraic closure
  $\overbar{\set K(x_1,\dots,x_{n-1})}$ of $\set K(x_1,\dots,x_{n-1})$
  and let $\omega\in\overbar{\set K}$ be a $\lambda_n$-th root of
  unity so that $\omega^{\lambda_n}=1$.  Since $r$ is nonzero, the
  complete factorization of $\bx{\lambda}-r$ over
  $\overbar{\set K(x_1,\dots,x_{n-1})}$ is given by
  \[
  \bx{\lambda}-r=x_1^{\lambda_1}\cdots\,x_{n-1}^{\lambda_{n-1}}
  \prod_{i=0}^{\lambda_n-1}\left(x_n-\omega^ir^{1/\lambda_n}
  x_1^{-\lambda_1/\lambda_n}\cdots\,x_{n-1}^{-\lambda_{n-1}/\lambda_n}\right).
  \]
  It then follows from the reducibility of $\bx{\lambda}-r$ over $\set
  K(x_1,\dots,x_{n-1})$ that there exist $i_1,\dots,i_k\in
  \{0,\dots,\lambda_n-1\}$ with $0<k<\lambda_n$ such that
  \[
  \prod_{j=1}^{k}\left(x_n-\omega^{i_j}r^{1/\lambda_n}
  x_1^{-\lambda_1/\lambda_n}\cdots\,x_{n-1}^{-\lambda_{n-1}/\lambda_n}\right)
  \in \set K(x_1,\dots,x_{n-1})[x_n].
  \]
  This implies that $(\lambda_i/\lambda_n)k\in \set Z$ for all
  $i=1,\dots,n-1$.  Thus $\lambda_n$ divides $k\cdot
  \gcd(\lambda_1,\dots,\lambda_{n-1})$ in~$\set Z$. Since
  $\lambda_1,\dots,\lambda_{n-1}$ are not all zero,
  $\gcd(\lambda_1,\dots,\lambda_n)=1$ and $\lambda_n>1$, we have
  $\lambda_n$ divides $k$ in~$\set Z$, a contradiction since
  $0<k<\lambda_n$.
  
  (ii) For the first equivalence, the sufficiency is evident. In order
  to show the necessity, suppose that $P(\bx{\lambda})$ is reducible
  over $\set K(x_1,\dots,x_{n-1})$.  Let $f\in\set K(x_1,\dots,x_{n-1})[x_n]$
  be an irreducible factor of $P(\bx{\lambda})$. Then the degree of
  $f$ in $x_n$ is less than $\lambda_n\deg(P)$. By assertion (i),
  we obtain that there exists a nontrivial factor $g\in\R[y]$ dividing
  $P$ in $\R[y]$ and $\deg(g)=\deg_{x_n}(f)/\lambda_n<\deg(P)$, a
  contradiction with the assumption that $P$ is irreducible over $\R$.
  Therefore, $P(\bx{\lambda})$ is irreducible over~$\set K(x_1,\dots,x_{n-1})$.
	
  For the second equivalence, by Gau{\ss}' lemma, one easily sees that
  $P(\bx{\lambda})$ is irreducible over $\set K(x_1,\dots,x_{n-1})$ if
  and only if $\prim_{x_n}(P(\bx{\lambda}))$ is irreducible over
  $\R[x_1,\dots,x_{n-1}]$.  It thus amounts to showing the equivalence
  between the irreducibility of $\prim_{x_n}(P(\bx{\lambda}))$ over
  $\R[x_1,\dots,x_{n-1}]$ and its irreducibility over $\R$. The
  direction from $\R$ to $\R[x_1,\dots,x_{n-1}]$ is trivial. In order
  to see the converse, notice that any nontrivial factor of
  $\prim_{x_n}(P(\bx{\lambda}))$ can only belong to
  $\R[x_1,\dots,x_{n-1}]$ since $\prim_{x_n}(P(\bx{\lambda}))$ is
  irreducible over $\R[x_1,\dots,x_{n-1}]$. On the other hand, the
  existence of any such a nontrivial factor would contradict with the
  fact that $\prim_{x_n}(P(\bx{\lambda}))$ is primitive with respect
  to~$x_n$.  Accordingly, $\prim_{x_n}(P(\bx{\lambda}))$ must be
  irreducible over~$\R$.
\end{proof}

\begin{lemma}\label{LEM:qiltype}
  Let $p\in\R[\bm x]$ and $\bm{\lambda}\in\set Z^n$ with
  $\gcd(\lambda_1,\dots,\lambda_n)=1$, $\lambda_1,\dots,\lambda_{n-1}$
  not all zero and $\lambda_n>0$. Let $P\in\set K[y]$ be such that
  $P(0)\neq 0$ and $P(\bx{\lambda})$ divides $p$ in $\set
  K(x_1,\dots,x_{n-1})[x_n]$. Then $\bm{\lambda}$ is a $q$-integer
  linear type of $p$ with the corresponding univariate polynomial
  divided by $P$ in $\set K[y]$.
\end{lemma}
\begin{proof}
  Let $f\in\R[y]$ be a primitive irreducible factor of $P$. Since
  $P(0) \neq 0$, then $f$ is $q$-primitive. Notice that
  $\bm\lambda\in\set Z^n$ and $\lambda_n>0$.  So
  $\prim_{x_n}(f(\bx{\lambda}))=\bx{\alpha}f(\bx{\lambda})$ for some
  $\bm{\alpha}\in\set N^n$ with $\alpha_n=0$. This implies that
  $\prim_{x_n}(f(\bx{\lambda}))$ is a $q$-integer linear polynomial in
  $\R[\bm x]$ of type $\bm\lambda$.  Because $P(\bx{\lambda})$ divides
  $p$ in $\set K(x_1,\dots,x_{n-1})[x_n]$, so does $f(\bx{\lambda})$.
  One then concludes from Lemma~\ref{LEM:sub}~(ii) that
  $\prim_{x_n}(f(\bx{\lambda}))$ is an irreducible factor of $p$
  over~$\R$. Therefore, by Definition~\ref{DEF:qild}, $\bm{\lambda}$
  is a $q$-integer linear type of $p$ and $f$ divides its
  corresponding polynomial in $\R[y]$. Since $f$ is arbitrary, the
  lemma follows.
\end{proof}

We are now ready to prove Proposition~\ref{PROP:multitype2poly}.

\smallskip
\begin{proof}[Proof of Proposition~\ref{PROP:multitype2poly}]
Assume that $P^*\in \R[y]\setminus\R$ and let $f\in\set K[y]$ be a
monic irreducible factor of~$P^*$. Then $f(x_n)$ divides
$p(x_1^{\lambda_n},\dots,x_{n-1}^{\lambda_n},
x_nx_1^{-\lambda_1}\cdots x_{n-1}^{-\lambda_{n-1}})$ in $\set
K(x_1,\dots,x_{n-1})[x_n]$. Subsequently substituting $x_n$ by
$x_nx_1^{\lambda_1}\cdots x_{n-1}^{\lambda_{n-1}}$ and then $x_i$ by
$x_i^{1/\lambda_n}$ for $i = 1,\dots,n-1$ yields that
$f(x_nx_1^{\lambda_1/\lambda_n}\cdots
x_{n-1}^{\lambda_{n-1}/\lambda_n})$ divides $p$ in
$\overbar{\set K(x_1,\dots,x_{n-1})}[x_n]$ where
$\overbar{\set K(x_1,\dots,x_{n-1})}$ denotes the algebraic closure
of the field $\set K(x_1,\dots,x_{n-1})$. This implies that $f(y)\neq
y$, for, otherwise, we would have that $x_n$ divides $p$ in
$\overbar{\set K(x_1,\dots,x_{n-1})}[x_n]$ and then
$p(x_1,\dots,x_{n-1},0) = 0$, a contradiction with the primitivity of
$p$ with respect to $x_1$. Let $r\in\overbar{\set K}$ be a root of~$f$.
Then $r\neq 0$ and $f$ is its minimal polynomial in $\set K[y]$.  It
follows from the divisibility of $p$ by
$f(x_nx_1^{\lambda_1/\lambda_n}\cdots x_{n-1}^{\lambda_{n-1}/\lambda_n})$
that $p(x_1,\dots,x_{n-1},
rx_1^{-\lambda_1/\lambda_n}\cdots x_n^{-\lambda_{n-1}/\lambda_n}) = 0$.
Now consider the minimal polynomial $P\in\set K[y]$ of~$r^{\lambda_n}$.
By Lemma~\ref{LEM:sub}~(ii), $P(\bx{\lambda})$ is irreducible over
$\set K(x_1,\dots,x_{n-1})$.  Thus $P(\bx{\lambda})$, upon making it
monic with respect to $x_n$, gives rise to the minimal polynomial of
$x_n = rx_1^{-\lambda_1/\lambda_n}\cdots x_n^{-\lambda_{n-1}/\lambda_n}$.
Therefore, $P(\bx{\lambda})$ divides $p$ in $\set K(x_1,\dots,x_{n-1})[x_n]$.
One thus concludes from Lemma~\ref{LEM:qiltype} that $\bm \lambda$ is
a $q$-integer linear type of $p$, say $\bm\lambda = \bm\lambda_i$ for
some integer $i$ with $1\leq i \leq m$, and then $P$ divides $P_i$ in
$\set K[y]$. Notice that $f$ is the minimal polynomial of $r$ and
$P(r^{\lambda_n}) = 0$. So $f$ divides $P(y^{\lambda_n})$ and then
$P_i(y^{\lambda_n})$ in~$\set K[y]$. As $f$ is arbitrary, we
have that $P^*$ divides $P_i(y^{\lambda_n})$ in $\set K[y]$.  Since
both polynomials are $q$-primitive and $\lambda_n>0$, then $P^*$
divides $P_i(y^{\lambda_n})$ in $\R[y]$ by Gau\ss' lemma.

In order to show the first assertion, it remains to verify that
$P_i(y^{\lambda_n})$ divides $P^*$ in $\R[y]$, and then $P^*$ and
$P_i(y^{\lambda_n})$ only differ by a unit in $\R$, yielding the
assertion.

Since $\bm\lambda = \bm\lambda_i$, by a simple calculation, one sees
from \eqref{EQ:multiqild} that $P_i(y^{\lambda_n})$ divides all
coefficients of $p(x_1^{\lambda_n},\dots,x_{n-1}^{\lambda_n},
yx_1^{-\lambda_1}\cdots x_{n-1}^{-\lambda_{n-1}})$ with respect to
$x_1,\dots,x_{n-1}$. By the definition of $P^*$, we obtain that
$P_i(y^{\lambda_n})$ divides $P^*$ in $\R[y]$. This actually also
shows that $P^*\notin \R$ if $\bm\lambda=\bm\lambda_i$ is a
$q$-integer linear type of $p$, because $P_i\notin \R$. This completes
the proof of the second assertion.
\end{proof}

\subsection{Algorithm and example}
Assembling everything together yields our first approach.

\smallskip\noindent{\bf MultivariateQILD$_1$.} Given a polynomial
$p\in\R[\bm x]$, compute its $q$-integer linear decomposition.

\begin{enumerate}
\item If $p \in \R$ then set $c = p$; and return $c$.

  \smallskip
\item Set $c = \cont(p)$ and $f = \prim(p)$. If $\supp(f)$ is a
  singleton then set $\bm\alpha$ to be the only element and update
  $c=cf/\bx{\alpha}$; and return $c\bx{\alpha}$.

  \smallskip
\item If $n=1$ then set $\alpha_1$ to be the lowest degree of $f$ with
  respect to $x_1$, $m=1$, $\lambda_{m1}=1$ and
  $P_m(y)=f(y)/y^{\alpha_1}$; and return
  $c\,x_1^{\alpha_1}\prod_{i=1}^mP_i(x_1^{\lambda_{i1}})$.

  \smallskip
\item Set $\bm\alpha=\bm 0$, $P_0=1$, $m=0$.\\[.5ex]
  For $i = 1,\dots,n$ do
  \begin{itemize}
  \item[4.1] Set $g = \cont_{x_i}(f)$, and update $f = \prim_{x_i}(f)$.
  \item[4.2] If $g \neq 1$ then call the algorithm recursively with
    input $g\in\R[x_1,\dots,x_{i-1},x_{i+1},\dots,x_n]$, returning
    \[
    g=x_1^{\tilde\alpha_1}\cdots\,x_{i-1}^{\tilde\alpha_{i-1}}
    x_{i+1}^{\tilde\alpha_{i+1}}\cdots\,x_n^{\tilde\alpha_n}\tilde
    P_0 \prod_{j=1}^{\tilde m}\tilde
    P_j(x_1^{\tilde\lambda_{j1}}\cdots\,x_{i-1}^{\tilde\lambda_{j,i-1}}
    x_{i+1}^{\tilde\lambda_{j,i+1}}\cdots\,x_n^{\tilde\lambda_{jn}}),
    \]
    update
    $\bm\alpha=\bm\alpha+(\tilde\alpha_1,\dots,\tilde\alpha_{i-1},
    0,\tilde\alpha_{i+1},\dots,\tilde\alpha_n)$, $P_0=P_0\tilde
    P_0$, and for $j=1,\dots, \tilde m$ iteratively update $m=m+1$,
    $\bm\lambda_m=(\tilde\lambda_{j1},\dots,\tilde\lambda_{j,i-1},0,
    \tilde\lambda_{j,i+1},\dots,\tilde\lambda_{jn})$,
    $P_m(y)=\tilde P_j(y)$.
  \end{itemize}

  \smallskip
\item If $\deg(f)=0$ then update $c=cf$; and return $c\,
  \bx{\alpha}P_0 \prod_{i=1}^{m}P_i(\bx{\lambda_i})$.

  \smallskip
\item Find the multiset $\Lambda$ of direction vectors of edges of
  $\newt(f)$ having no zero coordinates.

  \smallskip
\item If $\Lambda$ has more than one element then
\begin{itemize}
\item[7.1] Update $\Lambda$ to be its subset composed of elements with
  multiple occurrences.

\item[7.2] For fixed $\bm v\in\supp(f)$, find the set $\tilde\Lambda$
  consisting of direction vectors of line segments connecting $\bm v$
  and all other points in $\supp(p)$ which have no zero coordinates.

\item[7.3] Update $\Lambda$ to be $\Lambda\cap \tilde \Lambda$.
\end{itemize}
  
\item For $\bm \lambda$ in $\Lambda$ do
  \begin{itemize}
  \item[8.1] Set $P^*(y)$ to be the content of the numerator of
    $f(x_1^{\lambda_n},\dots,x_{n-1}^{\lambda_n},
    yx_1^{-\lambda_1}\cdots\,x_{n-1}^{-\lambda_{n-1}})$ with respect to
    $x_1,\dots,x_{n-1}$.
  \item[8.2] If $\deg(P^*)>0$ then
    \begin{quote}
      Update $m=m+1$, $\bm\lambda_m =\bm\lambda$,
      $P_m(y)=P^*(y^{1/\lambda_n})$.\\
      Set $f^*,g^*\in\R[x_1,\dots,x_n]$ to be the numerator and
      denominator of $P_m(\bx{\lambda})$, and update
      $f = f/f^*$ and
      $\alpha_i=\alpha_i+\deg_{x_i}(g^*)$ for $i=1,\dots,n-1$.
    \end{quote}
  \end{itemize}

\item If $\deg(f)>0$ then update $P_0=P_0 f$ else update $c=cf$.

  \smallskip
\item Return $c\, \bx{\alpha}P_0 \prod_{i=1}^{m}P_i(\bx{\lambda_i})$.
\end{enumerate}

\begin{theorem}\label{THM:multiqild1}
  Let $p\in\R[\bm x]$. Then the algorithm {\bf MultivariateQILD$_1$} 
  terminates and correctly computes the $q$-integer linear 
  decomposition of~$p$.
\end{theorem}
\begin{proof}
  This is evident by Propositions~\ref{PROP:multiqiltypes} and
  \ref{PROP:multitype2poly}.
\end{proof}

\begin{remark}\label{REM:multinonqil}
  If one is merely interested in only determining the $q$-integer
  linearity of the input polynomial $p\in\R[\bm x]$, rather than the
  full $q$-integer linear decomposition, then the above algorithm can
  be easily modified: any of the following conditions will trigger the
  adapted algorithm to terminate early, returning that $p$ is not
  $q$-integer linear.
  \begin{itemize}
  \item In Step~4.2, the polynomial $g$ turns out to be
    non-$q$-integer linear.
  
  \item (Proposition~\ref{PROP:multinecessity}) In Step~6, the Newton
    polytope of $f$ is not a zonotope; or there exists an edge of
    $\newt(f)$ whose direction vector has zero coordinates. In
    particular, the support $\supp(f)$ has more than one element whose
    certain coordinate attains the extremum value.

  \item (Proposition~\ref{PROP:multitype2poly}) In Step~8.2, the case
    of $\deg(P^*)=0$ happens, that is, the candidate $\bm \lambda$
    currently under investigation is fake.
    
  \item (Definition~\ref{DEF:qild}) In Step~10, we have $\deg(P_0)>0$.
  \end{itemize}
\end{remark}

\begin{example}\label{EX:multiqild1}
  Consider the polynomial $p\in\set Z[q,q^{-1}][x_1,x_2,x_3,x_4]$ of
  the form
  \begin{align}
  p &= 2q^2x_1^9x_2^{12}x_3^{13}+2qx_1^8x_2^{14}x_3^{13}
  +2qx_1^8x_2^{14}x_3^{12}x_4+18q^2x_1^{11}x_2^8x_3^{16}x_4^5
  +18qx_1^{10}x_2^{10}x_3^{16}x_4^5
  \nonumber\\
  &
  +18qx_1^{10}x_2^{10}x_3^{15}x_4^6-2qx_1^5x_2^{20}x_3^7x_4^7
  -2x_1^4x_2^{22}x_3^7x_4^7-2x_1^4x_2^{22}x_3^6x_4^8
  -18qx_1^7x_2^{16}x_3^{10}x_4^{12}
  \nonumber\\
  &
  -18x_1^6x_2^{18}x_3^{10}x_4^{12}-18x_1^6x_2^{18}x_3^9x_4^{13}
  +7q^2x_1x_2^{28}x_3x_4^{14}+7qx_2^{30}x_3x_4^{14}
  +7qx_2^{30}x_4^{15}
  \nonumber\\
  &
  +6q^4x_1^{15}x_3^{22}x_4^{15}+6q^3x_1^{14}x_2^2x_3^{22}x_4^{15}
  +6q^3x_1^{14}x_2^2x_3^{21}x_4^{16}+63q^2x_1^3x_2^{24}x_3^4x_4^{19}
  \nonumber\\
  &
  +63qx_1^2x_2^{26}x_3^4x_4^{19} +63qx_1^2x_2^{26}x_3^3x_4^{20}
  -6q^3x_1^{11}x_2^8x_3^{16}x_4^{22}-6q^2x_1^{10}x_2^{10}x_3^{16}x_4^{22}
  \nonumber\\
  &
  -6q^2x_1^{10}x_2^{10}x_3^{15}x_4^{23}+21q^4x_1^7x_2^{16}x_3^{10}x_4^{29}
  +21q^3x_1^6x_2^{18}x_3^{10}x_4^{29}+21q^3x_1^6x_2^{18}x_3^9x_4^{30}
  \label{EQ:example}
  \end{align}
  In order to compute the $q$-integer linear decomposition of the
  polynomial $p$ over $\set Z[q,q^{-1}]$, the algorithm
  {\bf MultivariateQILD$_1$} first tries to find candidates for all
  possible $q$-integer linear types of $p$. In this respect, it
  computes the Newton polytope of $p$ from its support $\supp(p)$,
  which can be readily read out from \eqref{EQ:example}, and finds
  that $\newt(p)$ possesses 11 vertices:
  \begin{align*}
    &\Big\{\,
    v_0 := (9, 12, 13, 0), v_1 := (8, 14, 13, 0), v_2 := (8, 14, 12, 1),
    v_3 := (1, 28, 1, 14), \\
    &\ \  v_4 := (0, 30, 1, 14),v_5 := (0, 30, 0, 15), v_6 := (15, 0, 22, 15), 
   v_7 := (14, 2, 22, 15), \\
    &\ \ v_8 := (7, 16, 10, 29), v_9 := (6, 18, 10, 29),v_{10} := (6, 18, 9, 30)\,\Big.\Big\},
  \end{align*}
  and 19 edges:
   \begin{align*}
  \Big\{\,&[v_1,v_4], [v_4,v_9], [v_7, v_9], [v_1,v_7], [v_4,v_5], [v_1,v_2], 
  [v_2,v_5], [v_5,v_{10}], [v_9,v_{10}], [v_0,v_2], \\
  &[v_0,v_1], [v_6,v_7], [v_6,v_8], [v_8,v_9], [v_0,v_3], [v_3,v_5], [v_0,v_6],
  [v_3,v_8], [v_8,v_{10}]\,\Big\}.
  \end{align*}
  Based on Proposition~\ref{PROP:multiqiltypes} (namely Steps~6-7), one
  obtains three candidates for $q$-integer linear types of $p$, that is,
  $(-1,2,-1,1),\ (2,-4,3,5),\ (-4,8,-6,7)$.
  A subsequent content computation for each candidate finally leads to the
  following $q$-integer linear decomposition
  \begin{equation}\label{EQ:exampleqild}
    p = x_1^8x_2^{12}x_3^{12}\cdot P_0\cdot P_1(x_1^2x_2^{-4}x_3^3x_4^5)
    \cdot P_2(x_1^{-4}x_2^8x_3^{-6}x_4^7),
  \end{equation}
  where $P_0=qx_1x_3+x_2^2x_3+x_2^2x_4$, $P_1(y) = 3q^2y^3+qy+1$ and
  $P_2(y) =7qy^2-2y+2q$.
  
  Notice that there are two elements in the support $\supp(p)$ (namely
  the exponent vectors of the first two monomials in \eqref{EQ:example})
  attaining the minimum value of $x_4$. One thus immediately sees from
  Proposition~\ref{PROP:multinecessity} that the given polynomial $p$
  is not $q$-integer linear. Also, the candidate $(-1,2,-1,1)$ turns
  out to be fake, implying, once again, the non-$q$-integer linearity
  of~$p$.
\end{example}

\section{$q$-Integer linear decomposition: the second approach}
\label{SEC:2ndapproach}
In this section we present our second approach for computing the
$q$-integer linear decomposition of a polynomial in an arbitrary
number of variables. This approach uses a bivariate-based
scheme, where the base bivariate case is tackled by the first approach 
from the preceding section. In order to describe it concisely, we need 
a $q$-analogue of \citep[Proposition~7]{AbPe2002a}. To this end, we 
require two technical lemmas. The first one corresponds to
\citep[Lemma~2]{AbPe2002a} but restricted to the case of Laurent
polynomials.
\begin{lemma}\label{LEM:qconstant}
  Let $p \in \R[x,x^{-1}]$ be a nonzero Laurent polynomial. If there
  exists a nonzero integer $a$ and a nonzero element $c\in \R$ such
  that $p(q^a x) = c p(x)$, then $c=q^{a m}$ for some $m\in \set Z$
  and $p(x)/x^m\in \R$.
\end{lemma}
\begin{proof}
  The assertion is clear if $p$ has only one monomial. Otherwise, let
  $x^i$ and $x^j$ with $i,j\in\set Z$ be two monomials of $p$.
  Extracting their coefficients in the identity $p(q^ax) = cp(x)$
  gives $q^{ai} = c = q^{aj}$. Thus $c$ has the form $q^{ai}$ for some
  $i\in\set Z$ and all the exponents $j$ of the monomials in $p$
  satisfy $a(j-i) = 0$, yielding $j = i$ as $a$ is nonzero. The lemma
  follows.
\end{proof}
Evidently, the above lemma remains valid by replacing the ring $\R$
with any of its ring extensions which is independent of the variable
$x$, or changing the variable $x$ to any its rational power $x^r$ for
$r\in\set Q$.  The next lemma plays the role of \citep[Lemma~3]{AbPe2002a}
in the $q$-shift setting, which describes a nice structure of
$q$-shift invariant bivariate polynomials.
\begin{lemma}\label{LEM:biqinvariant}
  Let $p\in \R[x,y]$. If there exists $c\in \R$ and $a,b\in \set Z$,
  not both zero, such that $p(q^ax,q^by)=cp(x,y)$, then there is a
  univariate polynomial $P\in \R[y]$ and four integers
  $\alpha,\beta,\lambda,\mu$ with $\lambda,\mu$ not both zero such
  that $p=x^\alpha y^\beta P(x^\lambda y^\mu)$.
\end{lemma}
\begin{proof}
  Without loss of generality, we assume that $a$ is nonzero. Otherwise,
  we can switch the roles of $x$ and $y$ in the following proof.
  Define $h(x,y) = p(x,yx^{b/a})$. Then $h \in \R[x^{1/a},x^{-1/a},y]$ and
  $p(x,y) = h(x,yx^{-b/a})$. Using $p(q^ax,q^by)=cp(x,y)$, a simple 
  calculation shows that $h(q^ax,y) = p(q^ax,q^byx^{b/a}) = ch(x,y)$.
  Viewing $h$ as a Laurent polynomial in $x^{1/a}$ over $\R[y]$, 
  Lemma~\ref{LEM:qconstant} implies that $h/x^{m/a}\in \R[y]$ for 
  some $m\in\set Z$. From the definition of $h$ we have that
  $i + (b/a)j = m/a$ for all $(i,j) \in\supp(p)$. Let $x^\alpha y^\beta$
  with $\alpha,\beta\in\set N$ be the trailing monomial in $p$, and 
  let $\lambda,\mu\in\set Z$ be such that $\lambda/\mu=-b/a$,
  $\gcd(\lambda,\mu) = 1$ and $\mu>0$. Then
  $\mu(i-\alpha) = \lambda (j-\beta)$ for all $(i,j) \in\supp(p)$. 
  By the coprimeness of $\lambda$ and $\mu$, one obtains that 
  for any $(i,j)\in\supp(p)$, there exists $k\in\set N$
  such that $(i,j) = (\alpha,\beta) + k(\lambda,\mu)$. 
  It thus follows that $p = x^\alpha y^\beta P(x^\lambda y^\mu)$
  for some $P\in\R[y]$. 
\end{proof}

From the above lemma, we are then able to establish the fact that the
problem of multivariate $q$-integer linearity is made up of a
collection of subproblems of bivariate $q$-integer linearity.
\begin{proposition}\label{PROP:multi2bi}
  Let $p\in \R[\bm x]$. Then there exists a univariate
  polynomial $P\in \R[y]$ and two vectors $\bm\alpha\in\set N^n$,
  $\bm\lambda\in\set Z^n\setminus\{\bm 0\}$ such that $p=\bx{\alpha}
  P(\bx{\lambda})$ if and only if for each pair $(i,j)$ with $1\leq
  i<j\leq n$, there is a polynomial $P_{ij}(y)\in
  \R[x_1,\dots,x_{i-1},x_{i+1},\dots,x_{j-1},x_{j+1},\dots,x_n][y]$
  and four integers $\beta_{ij},\beta_{ji},\mu_{ij},\mu_{ji}$ with
  $\mu_{ij},\mu_{ji}$ not both zero such that
  $p=x_i^{\beta_{ij}}x_j^{\beta_{ji}}P_{ij}(x_i^{\mu_{ij}}x_j^{\mu_{ji}})$.
\end{proposition}
\begin{proof}
  The necessity is clear. For the sufficiency, we proceed by induction
  on the number $n$ of variables. There is nothing to show in the base
  case where $n=1$. Assume that $n>1$ and the assertion holds for~$n-1$.

  Consider $p$ as a polynomial in $x_1,\dots,x_{n-1}$ over
  $\R[x_n]$. By the induction hypothesis, there is a polynomial
  $P^*(y)\in \R[x_n][y]$ and two vectors
  $(\alpha_1^*,\dots,\alpha_{n-1}^*)\in\set N^{n-1}$,
  $(\lambda_1^*,\dots,\lambda_{n-1}^*)\in\set Z^{n-1}$ with the
  $\lambda_i^*$ not all zero such that
  \[
  p(x_n)(x_1,\dots,x_{n-1})=x_1^{\alpha_1^*}\cdots\,x_{n-1}^{\alpha_{n-1}^*}
  P^*(x_1^{\lambda_1^*}\cdots\,x_{n-1}^{\lambda_{n-1}^*}).
  \]
  We may assume without loss of generality that $\lambda_1^*\neq 0$.
  Regarding $P^*$ as an element of $\R[y,x_n]$, we rewrite the
  preceding equation as
  \begin{equation}\label{EQ:n-1}
  p(x_1,\dots,x_n)=x_1^{\alpha_1^*}\cdots\, x_{n-1}^{\alpha_{n-1}^*}
  P^*(x_1^{\lambda_1^*}\cdots\,x_{n-1}^{\lambda_{n-1}^*},x_n).
  \end{equation}
  By taking $i=1$ and $j=n$ in the assumption, we know that
  $p=x_1^{\beta_{1n}}x_n^{\beta_{n1}}P_{1n}(x_1^{\mu_{1n}}x_n^{\mu_{n1}})$
  for $P_{1n}\in \R[x_2,\dots,x_{n-1}][y]$ and
  $\beta_{1n},\beta_{n1},\mu_{1n},\mu_{n1}\in \set Z$ with
  $\mu_{1n},\mu_{n1}$ not both zero. Therefore,
  \[
  p(q^{\mu_{n1}}x_1,x_2,\dots,x_{n-1},q^{-\mu_{1n}}x_n)
  =cp(x_1,\dots,x_n) \quad\text{with}\
  c = q^{\beta_{1n}\mu_{n1}-\beta_{n1}\mu_{1n}}\in \R.
  \]
  It follows from \eqref{EQ:n-1} that
  $P^*(q^{\mu_{n1}\lambda_1^*}x_1^{\lambda_1^*}\cdots\,
  x_{n-1}^{\lambda_{n-1}^*},q^{-\mu_{1n}}x_n)
  =cq^{-\mu_{n1}\alpha_1^*}P^*(x_1^{\lambda_1^*}\cdots\,x_{n-1}^{\lambda_{n-1}^*},x_n)$,
  that is,
  \[  
  P^*(q^{\mu_{n1}\lambda_1^*}y,q^{-\mu_{1n}}x_n)=c q^{-\mu_{n1}\alpha_1^*}P^*(y,x_n).
  \]
  Applying Lemma~\ref{LEM:biqinvariant} to $P^*(y,x_n)$ yields that
  there is a univariate polynomial $P\in \R[y]$ and four integers
  $\alpha_n,\alpha_n^*,\lambda_n,\lambda_n^*$ with
  $\lambda_n,\lambda_n^*$ not both zero such that
  $P^*(y,x_n)=y^{\alpha_n^*}x_n^{\alpha_n}P(y^{\lambda_n^*}x_n^{\lambda_n})$.
  Substituting $y=x_1^{\lambda_1^*}\cdots\,x_{n-1}^{\lambda_{n-1}^*}$
  into this equation, together with \eqref{EQ:n-1}, implies that
  $p=\bx{\alpha}P(\bx{\lambda})$ with
  $\bm\alpha=(\alpha_1^*+\lambda_1^*\alpha_n^*,\dots,
  \alpha_{n-1}^*+\lambda_{n-1}^*\alpha_n^*,\alpha_n)$ and
  $\bm\lambda=(\lambda_1^*\lambda_n^*,\dots,\lambda_{n-1}^*\lambda_n^*,
  \lambda_n)$. The proof follows by noticing that $\bm\lambda\neq\bm 0$.
\end{proof}

Inspired by the above proposition, we propose an algorithm which takes
a multivariate polynomial as input and computes its $q$-integer linear
decomposition in an iterative fashion.  At each iteration step, only
two variables are used with the others treated as coefficient
parameters.

\smallskip\noindent{\bf MultivariateQILD$_2$.} Given a polynomial
$p\in\R[\bm x]$, compute its $q$-integer linear decomposition.

\begin{enumerate}
\item If $p \in \R$ then set $c = p$; and return $c$.

  \smallskip
\item Set $c = \cont(p)$ and $f = \prim(p)$. If $\supp(f)$ is a
  singleton then set $\bm\alpha$ to be the only element and update
  $c=cf/\bx{\alpha}$; and return $c\bx{\alpha}$.

  \smallskip
\item If $n=1$ then set $\alpha_1$ to be the lowest degree of $f$ with
  respect to $x_1$, $m=1$, $\lambda_{m1}=1$ and
  $P_m(y)=f(y)/y^{\alpha_1}$; and return
  $c\,x_1^{\alpha_1}\prod_{i=1}^mP_i(x_1^{\lambda_{i1}})$.

  \smallskip
\item If $n=2$ then call the algorithm {\bf MultivariateQILD$_1$} with
  input $f\in \R[x_1,x_2]$ to compute its $q$-integer linear
  decomposition
  \[
  f=x_1^{\alpha_1}x_2^{\alpha_2}P_0
  \prod_{i=1}^{m}P_i(x_1^{\lambda_{i1}}x_2^{\lambda_{i2}});
  \] 
  and then return $c\, x_1^{\alpha_1}x_2^{\alpha_2}P_0
  \prod_{i=1}^{m}P_i(x_1^{\lambda_{i1}}x_2^{\lambda_{i2}})$.

  \smallskip
\item Set $\bm\alpha=\bm 0$, $P_0=1$, $m=0$ and $g =
  \cont_{x_1,x_2}(f)$, and update $f = \prim_{x_1,x_2}(f)$.

  \smallskip
\item If $g \neq 1$ then call the algorithm recursively with input
  $g\in \R[x_3,\dots,x_n]$, returning
  \[
  g=x_3^{\tilde\alpha_3}\cdots\,x_n^{\tilde\alpha_n}\tilde P_0
  \prod_{i=1}^{\tilde m}\tilde
  P_i(x_3^{\tilde\lambda_{i3}}\cdots\,x_n^{\tilde\lambda_{in}}),
  \]
  update $\bm\alpha=\bm\alpha+(0,0,\tilde \alpha_3,\dots,\tilde
  \alpha_n)$, $P_0=P_0 \tilde P_0$, and for $i=1,\dots, \tilde m$
  iteratively update $m=m+1$, $\bm\lambda_m=(0,0,\tilde
  \lambda_{i3},\dots,\tilde\lambda_{in})$, $P_m(y)=\tilde P_i(y)$.

  \smallskip
\item If $\supp(f)$ is a singleton then set $\bm\alpha^*$ to be the
  only element and update $\bm\alpha=\bm\alpha+\bm\alpha^*$,
  $c=cf/\bx{\alpha^*}$; and return $c\, \bx{\alpha}P_0
  \prod_{i=1}^{m}P_i(\bx{\lambda_i})$.

  \smallskip
\item Set $\Lambda_1 = \{\big((1),f(y,x_2,\dots,x_n)\big)\}$.\\[.5ex]  
  For $k=1,\dots,n-1$ do
  \begin{itemize}
  \item[8.1] Set $\Lambda_{k+1}=\{\}$.
  \item[8.2] For
    $\big((\mu_1,\dots,\mu_k),h(y,x_{k+1},\dots,x_n)\big)$ in
    $\Lambda_k$ do
  \begin{quote}
    Call the algorithm {\bf MultivariateQILD$_1$} with input $h\in
    \R[x_{k+2},\dots,x_n][y, x_{k+1}]$ to compute its $q$-integer
    linear decomposition
    \begin{equation}\label{EQ:hstep8}
      h = y^{\alpha^*}x_{k+1}^{\beta^*}P_0^*\prod_{i=1}^{m^*}
      P_i^*(y^{\lambda_{i}^* } x_{k+1}^{\mu_{i}^*},x_{k+2},\dots,x_n),
    \end{equation}
    where $P_0^*\in \R[y,x_{k+1},\dots,x_n]$ and
    $P_i^*(y,x_{k+2},\dots,x_n) \in \R[y,x_{k+2},\dots,x_n]$; then
    update $\bm\alpha$ by adding the vector
    $(\mu_1\alpha^*,\dots,\mu_k\alpha^*,\beta^*,0,\dots,0)$, update $P_0$
    by multiplying
    $P_0^*(x_1^{\mu_1}\cdots\,x_k^{\mu_k},x_{k+1},\dots,x_n)$ and update
    $\Lambda_{k+1}$ by joining the elements
    $\big((\mu_1\lambda_{i}^*,\dots,\mu_k\lambda_{i}^*,\mu_{i}^*),
    P_i^*(y,x_{k+2},\dots,x_n)\big)$ for $i=1,\dots,m^*$.
  \end{quote}
  \end{itemize}

  \smallskip
\item Set $g\in\R[\bm x]$ to be the denominator of $P_0$.
  Update $P_0$ to be its numerator, update $\alpha_i=\alpha_i-\deg_{x_i}(g)$
  for $i=1,\dots,n-1$, and for $\big(\bm \mu,h(y)\big)$ in
  $\Lambda_{n}$ iteratively update $m = m+1$, $\bm\lambda_m=\bm\mu$
  and $P_m(y) = h(y)$.

  \smallskip
\item Return $c\, \bx{\alpha}P_0 \prod_{i=1}^{m}P_i(\bx{\lambda_i})$.
\end{enumerate}

\begin{theorem}\label{THM:multiqild2}
  Let $p\in \R[\bm x]$. Then the algorithm {\bf
    MultivariateQILD$_2$} correctly computes the $q$-integer linear
  decomposition of $p$.
\end{theorem}
\begin{proof}
  The correctness immediately follows from Proposition~\ref{PROP:multi2bi}.
\end{proof}

\begin{example}\label{EX:multiqild2}
  Consider the same polynomial $p$ given by \eqref{EQ:example} as
  Example~\ref{EX:multiqild1}. In order to compute its $q$-integer
  linear decomposition over $\set Z[q,q^{-1}]$, the algorithm
  {\bf MultivariateQILD$_2$} (mainly Step~8) proceeds in the following
  three stages with their respective Newton polytopes plotted in
  Figure~\ref{FIG:newtonpoly2}. Firstly, by viewing $p$ as a
  polynomial in $x_1,x_2$ over $\set Z[q,q^{-1},x_3,x_4]$, applying
  the algorithm {\bf MultivariateQILD$_1$} to $p$ gives
  \begin{equation}\label{EQ:stage1}
    p = x_1^{15}P^{(1)}(x_1^{-1}x_2^2,x_3,x_4)
  \end{equation}
  with 
  \begin{align*}
  P^{(1)}(y,x_3,x_4)&=
  7qy^{15}x_3x_4^{14}+7qy^{15}x_4^{15}+7q^2y^{14}x_3x_4^{14}
  +63qy^{13}x_3^4x_4^{19}+63qy^{13}x_3^3x_4^{20}\\
  &
  +63q^2y^{12}x_3^4x_4^{19}-2y^{11}x_3^7x_4^7-2y^{11}x_3^6x_4^8
  -2qy^{10}x_3^7x_4^7+21q^3y^9x_3^{10}x_4^{29}-18y^9x_3^{10}x_4^{12}\\
  &
  +21q^3y^9x_3^9x_4^{30}-18y^9x_3^9x_4^{13}
  +21q^4y^8x_3^{10}x_4^{29}-18qy^8x_3^{10}x_4^{12}
  +2qy^7x_3^{13}+2qy^7x_3^{12}x_4\\
  &
  +2q^2y^6x_3^{13}-6q^2y^5x_3^{16}x_4^{22}+18qy^5x_3^{16}x_4^5
  -6q^2y^5x_3^{15}x_4^{23}+18qy^5x_3^{15}x_4^6-6q^3y^4x_3^{16}x_4^{22}
  \\
  &
  +18q^2y^4x_3^{16}x_4^5+6q^3yx_3^{22}x_4^{15}
  +6q^3yx_3^{21}x_4^{16}+6q^4x_3^{22}x_4^{15}.
  \end{align*}
  There is only one $q$-integer linear type, namely $(-1,2)$, of $p$
  over $\set Z[q,q^{-1},x_3,x_4]$. Next, with input
  $P^{(1)}(y,x_3,x_4)\in\set Z[q,q^{-1},x_4][y,x_3]$, calling the
  algorithm {\bf MultivariateQILD$_1$} again and substituting
  $y = x_1^{-1}x_2^2$ yields
  \begin{equation}\label{EQ:stage2}
    p = x_2^{28}\cdot P_0\cdot P^{(2)}(x_1^2x_2^{-4}x_3^3,x_4),
  \end{equation}
  where $P_0=qx_1x_3+x_2^2x_3+x_2^2x_4$ and
  $
  P^{(2)}(y,x_4) = 6q^3y^7x_4^{15}-6q^2y^5x_4^{22}+18qy^5x_4^5+2qy^4
  +21q^3y^3x_4^{29} -18y^3x_4^{12}-2y^2x_4^7+63qyx_4^{19}+7qx_4^{14}.
  $
  The vector $(2,-4,3)$ is then the only $q$-integer linear type of
  $p$ over $\set Z[q,q^{-1},x_4]$. Finally, the last call to the
  algorithm {\bf MultivariateQILD$_1$} with input
  $P^{(2)}(y,x_4)\in\set Z[q,q^{-1}][y,x_4]$, along with the
  substitution $y=x_1^2x_2^{-4}x_3^3$, leads to the desired
  decomposition \eqref{EQ:exampleqild}. The two $q$-integer linear
  types $(2,-4,3,5)$ and $(-4,8,-6,7)$ of $p$ over $\set Z[q,q^{-1}]$
  have been correctly recovered.

  From \eqref{EQ:stage1} and \eqref{EQ:stage2}, one sees that $p$ is
  $q$-integer linear over $\set Z[q,q^{-1},x_3,x_4]$ but it is not
  $q$-integer linear over $\set Z[q,q^{-1},x_4]$. This last point
  indicates the non-$q$-integer linearity of $p$ over $\set
  Z[q,q^{-1}]$, even before starting the third stage.
\end{example}

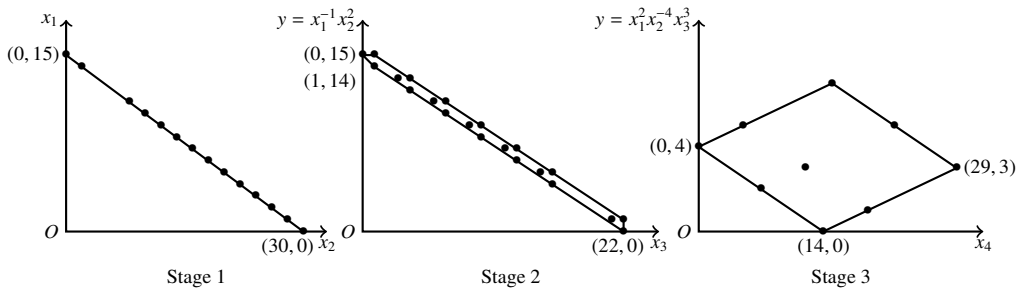
\begin{figure}[t]
    \centering
    \begin{tikzpicture}[scale=.13,every node/.style={scale=.75}]
      \begin{scope}[xscale=.8,yscale=1.2]
      \draw[<->,thick] (0,18) node[left]{$x_1$} -- (0,0) node[left]{$O$} -- (33,0) node[below]{$x_2$};
      \foreach \Point in {(30, 0), (28, 1), (26, 2), (24, 3), (22, 4),
        (20, 5), (18, 6), (16, 7), (14, 8), (12, 9), (10, 10), (8, 11),
        (2, 14), (0, 15)}{\node at \Point {\textbullet};
      }
  	  \node[left] at (0,15) {$(0,15)$};
  	  \node[below] at (28,0) {$(30,0)$};
      \draw[thick] (0,15) -- (30,0);
      \node at (16.5,-4) {Stage 1};
      \end{scope}
      \begin{scope}[xshift=30cm,xscale=1.2,yscale=1.2]
      \draw[<->,thick] (0,18) node[left]{$y=x_1^{-1}x_2^2$} -- (0,0) node[left]{$O$} -- (25,0) node[below]{$x_3$};
      \foreach \Point in {(22, 1), (22, 0), (21, 1), (16, 5), (16, 4),
        (15, 5), (13, 7), (13, 6), (12, 7), (10, 9), (10, 8), (9, 9), (7, 11),
        (7, 10), (6, 11), (4, 13), (4, 12), (3, 13), (1, 15), (1, 14), (0, 15)}{
        \node at \Point {\textbullet};
      }
      \node[left] at (0,15) {$(0,15)$};
      \node[below left] at (0,14) {$(1,14)$};
      \node[below] at (21.5,0) {$(22,0)$};
      \draw[thick] (0,15) -- (1,14) -- (22,0) -- (22,1) -- (1,15) -- (0,15);
      \node at (12.5,-4) {Stage 2};
      \end{scope}
      \begin{scope}[xshift=64cm,xscale=.9,yscale=2.16]
        \draw[<->,thick] (0,10) node[left]{$y=x_1^2x_2^{-4}x_3^3$} -- (0,0) node[left]{$O$} -- (32,0)
        node[below]{$x_4$};
      \foreach \Point in {(29, 3), (22, 5), (19, 1), (15, 7), (14, 0), (12, 3), (7, 2), (5, 5), (0, 4)}{\node at \Point {\textbullet};
      }
      \node[left] at (0,4) {$(0,4)$};
      \node[below] at (14,0) {$(14,0)$};
      \node[right] at (29,3) {$(29,3)$};
      \draw[thick] (0,4) -- (14,0) -- (29,3) -- (15,7) -- (0,4);
      \node at (16,-2.23) {Stage 3};
      \end{scope}
    \end{tikzpicture}
    \caption{Newton polytopes constructed in the three stages in
      Example~\ref{EX:multiqild2}.}\label{FIG:newtonpoly2}
\end{figure}

Once more, similar to Remark~\ref{REM:multinonqil}, the above
algorithm can be easily modified so as to determine the $q$-integer
linearity of a given polynomial only. In other words, the algorithm
can exit early and return a negative answer whenever one of the
following situations occurs.
\begin{itemize}
\item In Step~4 or in any iteration step of Step~8.2, any of the
  triggers listed in Remark~\ref{REM:multinonqil} is touched.
  
\item In Step~6, the polynomial $g$ turns out to be not $q$-integer
  linear.
\end{itemize}

\section{Complexity comparison}\label{SEC:complexity}
In this section, we give complexity analyses for the two algorithms
presented in Sections~\ref{SEC:1stapproach} and \ref{SEC:2ndapproach}
in the case of $\R = \set Z[q,q^{-1}]$.  In addition, we discuss two
more algorithms for the same purpose, namely for computing the
$q$-integer linear decomposition of polynomials, along with their
costs in the bivariate case for the sake of comparison.

\subsection{Complexity background}
We first collect some classical complexity notations and facts needed
in this paper. More background on these can be found in
\citep{vzGGe2013}.

Although our algorithms work in more general UFDs, we confine our
complexity analysis to the case of integer (Laurent) polynomials, that
is, when $\D$ is the ring of integers $\set Z$ and then $\R$ is equal
to $\set Z[q,q^{-1}]$. Here $q$ can be viewed as a variable in
addition to $x_1,\dots,x_n$. Note that operations in $\set Z[q,q^{-1}]$
can be easily transferred to those in $\set Z[q]$ with a negligible
cost. The cost is given in terms of number of word operations used so
that growth of coefficients comes into play. Recall that the
{\em word length} of a nonzero integer $a\in \set Z$ is defined as
$\bigO(\log|a|)$. In this paper, all complexity is analyzed in terms
of a function $\M(d)$ which bounds the cost required to multiply two
integers of word length at most $d$ or polynomials of degree at most
$d$. We take $\M(d) = d^2$ using classical arithmetic and $\M(d) =
\softO(d)$ using fast arithmetic, where the {\em soft-Oh notation}
\lq\lq $\softO$\rq\rq\ is basically \lq\lq $\bigO$\rq\rq\ but
suppressing logarithmic factors (see \cite[Definition~25.8]{vzGGe2013}
for a precise definition).  We assume that $\M$ is subadditive,
superlinear and subquadratic, that is, $\M(a)+\M(b)\leq \M(a+b)$ and
$a\M(b)\leq \M(ab) \leq a^2\M(b)$ for all $a,b\in\set N$.

Throughout this paper, we define the {\em max-norm} $||p||_\infty$ of
a Laurent polynomial $p\in \set Z[q,q^{-1}]$ as the maximum absolute
value of its coefficients with respect to $q$, and the {\em max-norm}
$||p||_\infty$ of a polynomial $p = \sum_{\bm i\in\set N}
p_{i_1,\dots,i_n}\bx{i}\in \set Z[q,q^{-1}][\bm x]$ as
$\max_{\bm i\in\set N}\{||p_{i_1,\dots,i_n}||_\infty\}$.
The GCD computation is fundamental for our algorithms. Before
analyzing the algorithm, let us recall some useful complexity results
on GCD computation.
\begin{lemma}[{\cite[Page 135-139]{Gelf1960}}]\label{LEM:facnorm}
  Let $p_1,\dots,p_m\in\set Z[\bm x]$. Let $p = p_1\cdots p_m$ and let
  $d_i=\deg_{x_i}(p)$ for all $i=1,\dots,n$. Then
  \[
  ||p_1||_\infty\cdots||p_m||_\infty\leq e^{d_1+\dots+d_n}||p||_\infty,
  \]
  where $e$ is the base of the natural logarithm. 
\end{lemma}
Note that when $n=1$ the above bound is actually worse than Mignotte's
factor bound for large $d$, which, however, leads to the same order of
magnitude for word lengths of the max-norms.

The lemma below provides bounds for the resultant of two multivariate
integer polynomials, which can be verified by following the proof of
\citep[Theorem~10]{BiLi2010} but arguing from the perspective of
multivariate polynomials.
\begin{lemma}\label{LEM:resultant}
  Let $f,g\in\set Z[\bm x]$ with $\deg_{x_i}(f),\deg_{x_i}(g)\leq d_i$
  for all $i = 1,\dots,n$.  Then
  \[
  ||\res_{x_n}(f,g)||_\infty\leq(2d_n)!(d_1+1)^{2d_n-1}\cdots
  (d_{n-1}+1)^{2d_n-1}||f||_\infty^{d_n}\,||g||_\infty^{d_n}.
  \]
\end{lemma}
The next result is likely known in the literature, but we could not find
a suitable reference, so we included a proof here for completeness.

\begin{lemma}\label{LEM:Gn}
  Let $f,g\in\set Z[\bm x]$ with $\deg_{x_i}(f)\leq d_i,\deg_{x_i}(g)\leq d_i$
  for all $i=1,\dots,n$, $||f||_\infty\leq \beta$ and $||g||_\infty
  \leq \beta$. Let $d=\max\{d_1,\dots,d_n\}$ and $D_n=d_1\cdots d_n$.
  Then computing $\gcd(f,g)$ over $\set Z$ takes
  $\bigO(D_n\M(nd+\log\beta)\log(nd+\log\beta))$ word operations.
\end{lemma}
\begin{proof}
  We proceed to compute $h=\gcd(f,g)$ by a small prime modular
  algorithm.  By Lemma~\ref{LEM:facnorm}, $||\gcd(f,g)||_\infty\leq
  e^{d_1+\dots+d_n} \beta\leq e^{nd}\beta=B$ with $e$ being the base
  of the natural logarithm.  Then $\log B\in\bigO(nd+\log\beta)$. Let
  $k = \lceil2\log_2((2d)!(d+1)^{(n-1)(2d-1)}\beta^{2d})\rceil$.  By
  Lemma~\ref{LEM:resultant}, the value $k$ is an upper bound on
  $2\log_2||\res_{x_n}(f/h,g/h)||_\infty$ and thus guarantees that at
  least $k/2$ of the first $k$ primes $p_1=2,\dots,p_k$ do not divide
  $\res_{x_n}(f/h,g/h)$. This means that at least half of the primes
  $p_1,\dots,p_k$ are \lq\lq lucky\rq\rq.  It is then sufficient to
  choose $\lceil\log_2(2B+1)\rceil\leq k/2$ \lq\lq lucky\rq\rq\ ones
  from these $k$ primes, each of word length $\bigO(\log k)$.  For
  every chosen prime $p$, we reduce all coefficients of $f$ and
  $g$ modulo $p$, using $\bigO(D_n\log\beta\log p)$ word operations,
  and compute $\gcd(f_p,g_p)$ with $f_p=f\mod p$ and $g_p=g\mod p$.
  The desired $\gcd(f,g)$ can be recovered by a final application of
  the Chinese remainder theorem, which takes
  $\bigO(D_n\M(nd+\log\beta)\log(nd+\log\beta))$ word operations.
  Neglecting the cost of computing primes, it remains to count the
  number of arithmetic operations, denoted by $G_p(n,d,D_n)$, used by
  the gcd computation in the field $\set Z_p$ for each prime $p$, with
  the rest following by the fact that each operation of these takes
  $\bigO(M(\log p))$ word operations and $\log p\in
  \bigO(\log(nd)+\log\log \beta)$.
	
  For each prime $p$, we compute $\gcd(f_p,g_p)$ with $f_p=f\mod p$
  and $g_p=g\mod p$ by an evaluation-interpolation scheme
  \citep{GCL1992}: evaluate coefficients of $f_p,g_p$ with respect to
  $x_1,\dots,x_{n-1}$ at $d_n$ points from $\set Z_p$ for $x_n$;
  compute $d_n$ GCDs over $\set Z_p$ of two $(n-1)$-variate
  polynomials of degrees at most $d_1,\dots,d_{n-1}$ in
  $x_1,\dots,x_{n-1}$, respectively; recover the final GCD by
  interpolation. Notice that there are at most $d_1\cdots
  d_{n-1}=D_n/d_n$ monomials in $x_1,\dots,x_{n-1}$ appearing in each
  of the polynomials $f_p$ and $g_p$. The process of evaluation and
  interpolation then takes $\bigO((D_n/d_n)\M(d_n)\log d_n)$
  arithmetic operations in the field $\set Z_p$.  The second step uses
  $\bigO(d_nG_p(n-1,d^{(n-1)}, D_{n-1}))$ arithmetic operations in
  $\set Z_p$, where $d^{(n-1)} = \max\{d_1,\dots,d_{n-1}\}$ and
  $D_{n-1}=d_1\cdots d_{n-1}$. Thus we obtain the recurrence relation
  \[
  \bigO(G_p(n,d,D_n))\subset \bigO((D_n/d_n)\M(d_n)\log d_n)
  +\bigO(d_nG_p(n-1,d^{(n-1)},D_{n-1})).
  \]
  From the initial condition that $G_p(1,d_1,d_1)$ is in
  $\bigO(\M(d_1)\log d_1)$, one concludes that $G_p(n,d,D_n)$ is in
  $\bigO((D_n/d)\M(d)\log D_n)$.
\end{proof}

\subsection{Cost analyses of our two algorithms}
We are now ready to present the cost of our first approach.  In order
to make it ready to use in the subsequent analysis of our second
approach, we analyze the cost in the case of
$\R = \set Z[q,q^{-1},z_1,\dots,z_v]$, where $v\in \set N$ is
arbitrary but fixed and the $z_i$ are additional parameters
independent of $q,x_1,\dots,x_n$.

\begin{theorem}\label{THM:multicost1}
  Let $p\in\set Z[q,q^{-1},z_1,\dots,z_v][\bm x]$. Assume that both
  the numerator and denominator of $p$ have maximum degree $d$ in each
  variable from $\{q,z_1,\dots,z_v,x_1,\dots,x_n\}$ separately, and
  let $||p||_\infty=\beta$. Then the algorithm {\bf MultivariateQILD$_1$}
  computes the $q$-integer linear decomposition of $p$ over
  $\set Z[q,q^{-1},z_1,\dots,z_v]$ using
  $$\bigO(n!d^{2n+v+2}\M((n^3+nv)d+n\log\beta)
  \log((n^2+v)d+\log\beta)+n!d^{n\lfloor n/2\rfloor}\M(n\log d)\log\log d)$$
  word operations.
\end{theorem}
\begin{proof}
  Let $T(n,d,\log\beta)$ denote the number of word operations used by
  the algorithm applied to the polynomial~$p$. Steps~1 and 5 treat the
  trivial case, taking no word operations. In Step~2, finding the
  content $c$ amounts to computing a GCD of at most $(d+1)^n$
  polynomials in $\set Z[q,z_1,\dots,z_v]$ of degree at most $d$ in
  each variable separately and max-norm at most $\beta$. Thus by
  Lemma~\ref{LEM:Gn}, this step takes
  $\bigO(d^{n+v+1}\M((v+1)d+\log\beta)\log((v+1)d+\log\beta))$ word
  operations. Step~3 deals with the univariate case, yielding that the
  initial cost $T(1,d,\log\beta)$ is in
  $\bigO(d^{v+2}\M((v+1)d+\log\beta)\log((v+1)d+\log\beta))$.

  In Step~4, at each iteration of the loop, the computation of the
  content $g$ and its primitive part in Step~4.1 can be done using
  $\bigO(d^{n+v+1}\M((n+v)d+\log\beta)\log((n+v)d+\log\beta))$; while
  Step~4.2 takes $\bigO(T(n-1,d,nd+\log\beta))$ word operations as
  $g\in\set Z[q,x_1,\dots,x_{i-1},x_{i+1},\dots,x_n]$ of maximum
  degree at most $d$ in each variable separately and max-norm of word
  length $\bigO(nd+\log\beta)$ by Lemma~\ref{LEM:facnorm}. Since there
  are $n$ iterations, this step in total takes
  $\bigO(nd^{n+v+1}\M((n+v)d+\log\beta)\log((n+v)d+\log\beta))
  +\bigO(nT(n-1,d,nd+\log\beta))$ word operations.
  
  The computation of the Newton polytope of $f$ dominates the other
  costs in Steps~6-7, which, by \cite[Theorem~26.3.1]{GOT2018}, takes
  $\bigO((s\log s+s^{\lfloor n/2\rfloor})\M(\log d)\log\log d)$ word
  operations with $s$ denoting the cardinality of $\supp(f)$.  Since
  $s\leq (d+1)^n$, we obtain the total cost $\bigO((nd^n\log
  d+d^{n\lfloor n/2\rfloor})\M(\log d)\log\log d)$ for Steps~6-7. In
  Step~8, for each $\bm\lambda\in\Lambda$, a direct calculation shows
  that $f(x_1^{\lambda_n},\dots,x_{n-1}^{\lambda_n},
  yx_1^{-\lambda_1}\cdots\,x_{n-1}^{-\lambda_{n-1}})$ has degree in
  $y$ at most $d$, max-norm of word length $\bigO(nd+\log\beta)$ and
  at most $(d+1)^n$ nonzero monomials in $x_1,\dots,x_{n-1}$
  appearing. Thus by Lemma~\ref{LEM:Gn}, Step~8.1 takes
  $\bigO(d^{n+v+2}\M((n+v+2)d+\log\beta)\log((n+v+2)d+\log\beta))$
  word operations, which dominates the cost for Step~8.2. Since there
  are at most $s-1\leq (d+1)^n-1$ elements in the set $\Lambda$, this
  step takes
  $\bigO(d^{2n+v+2}\M((n+v+2)d+\log\beta)\log((n+v+2)d+\log\beta))$
  word operations. Steps~9 and 10 both take no word operations without
  expanding the product.

  In summary, we obtain the recurrence relation
  \begin{align*}
    \bigO(T(n,d,\log\beta))&\subset 
    \bigO(d^{2n+v+2}\M((n+v+2)d+\log\beta)\log((n+v+2)d+\log\beta)\\
    &\quad+d^{n\lfloor n/2\rfloor}\M(\log d)\log\log d)
    +\bigO(nT(n-1,d,nd+\log\beta)),
  \end{align*}
  along with
  $T(1,d,\log\beta)\in\bigO(d^{v+2}\M((v+1)d+\log\beta)\log((v+1)d+\log\beta))$.
  The cost follows.
\end{proof}
\begin{corollary}\label{COR:multicost1}
  With the assumptions of Theorem~\ref{THM:multicost1}, further let
  $v = 0$. Then the algorithm {\bf MultivariateQILD$_1$} computes the
  $q$-integer linear decomposition of $p$ over $\set Z[q,q^{-1}]$
  using $\softO(n!d^{2n+4}+d^{2n+2}\log^2\beta+n!d^{n\lfloor n/2\rfloor})$
  word operations with classical arithmetic and
  $\softO(n!d^{2n+3}+n!d^{2n+2}\log\beta+n!d^{n\lfloor n/2\rfloor})$
  with fast arithmetic.
\end{corollary}

In the case of our second algorithm we have the following cost.
\begin{theorem}\label{THM:multicost2}
  Let $p\in \set Z[q,q^{-1}][\bm x]$. Assume that both the numerator
  and denominator of $p$ have maximum degree $d$ in each variable from
  $\{q,x_1,\dots,x_n\}$ separately, and let $||p||_\infty = \beta$.
  Then the algorithm {\bf MultivariateQILD$_2$} computes the
  $q$-integer linear decomposition of $p$ over $\set Z$ using
  $\bigO(d^{n+4}\M(n^4d+n^2\log\beta)\log(n^2d+\log\beta))$ word
  operations.
\end{theorem}
\begin{proof}
  Let $T(n,d,\log\beta)$ denote the number of word operations used by
  the algorithm applied to the polynomial $p$. The first three steps
  are exactly the same as the algorithm {\bf MultivariateQILD$_1$}.
  Thus, as before, Step~1 takes no word operations, Step~2 uses
  $\bigO(d^{n+1}\M(d+\log\beta)\log(d+\log\beta))$ word operations,
  and Step~3 gives the initial cost $T(1,d,\log\beta)\in
  \bigO(d^2\M(d+\log\beta)\log(d+\log\beta))$.  Step~4 deals with the
  bivariate case. By Theorem~\ref{THM:multicost1} with $n=2$ and
  $v=0$, this step yields that $T(2,d,\log\beta)$ is in
  $\bigO(d^6\M(d+\log\beta)\log(d+\log\beta))$.

  In Step~5, by Lemma~\ref{LEM:Gn}, the computation of the content and
  primitive part can be done within
  $\bigO(d^{n+1}\M(nd+\log\beta)\log(nd+\log\beta))$ word operations.
  Notice that $g\in \set Z[q,x_3,\dots,x_n]$ has maximum degree at
  most $d$ in each variable separately and max-norm of word length
  $\bigO(nd+\log\beta)$ by Lemma~\ref{LEM:facnorm}.  Then Step~6 takes
  $\bigO(T(n-2,d,nd+\log\beta))$ word operations. Step~7 takes linear
  time in the cardinality of $\supp(f)$, which is at most
  $(d+1)^n$. In Step~8, notice that for the $k$th iteration, the
  polynomial $h\in\set Z[q,x_{k+2},\dots,x_n][y,x_{k+1}]$ has maximum
  degree at most $d$ in each variable separately and max-norm of word
  length $\bigO(nd+\log\beta)$. Thus by Theorem~\ref{THM:multicost1}
  with $n=2$ and $v=n-k-1$, the $k$th iteration requires
  $\bigO(d^{n-k+5}\M((n-k-1)d+\log\beta)\log((n-k-1)d+\log\beta))$
  word operations. Since $1\leq k\leq n-1$, this step in total takes
  $\bigO(d^{n+4}\M(n^2d+n\log\beta)\log(nd+\log\beta))$ word
  operations, dominating the costs of Steps~9-10.

  In summary, we obtain the recurrence relation
  \[
  \bigO(T(n,d,\log\beta))\subset
  \bigO(d^{n+4}\M(n^2d+n\log\beta)\log(nd+\log\beta))
  +\bigO(T(n-2,d,nd+\log\beta)),
  \]
  along with
  $T(1,d,\log\beta)\in\bigO(d^2\M(d+\log\beta)\log(d+\log\beta))$ and
  $T(2,d,\log\beta)\in\bigO(d^6\M(d+\log\beta)\log(d+\log\beta))$. The
  announced cost follows.
\end{proof}
\begin{corollary}\label{COR:multicost2}
  With the assumptions of Theorem~\ref{THM:multicost2}, the algorithm
  {\bf MultivariateQILD$_2$} computes the $q$-integer linear
  decomposition of $p$ over $\set Z[q,q^{-1}]$ using
  $\bigO(d^{n+6}+d^{n+4}\log^2\beta)$ word operations with classical
  arithmetic and $\softO(d^{n+5}+d^{n+4}\log\beta)$ with fast
  arithmetic.
\end{corollary}
\begin{remark}\label{REM:conflitti}
The complexity of our both approaches could be further improved if one
finds a multivariate version of the GCD algorithm of \cite{conf2003}.
This is the algorithm which randomly reduces computing the GCD of
several polynomials over a finite field to computing a single GCD of
two polynomials over the same field.
\end{remark}

\subsection{Cost analysis of the resultant-based algorithm}
In this subsection, we review the algorithm of \cite{Le2001}.  As
mentioned in the introduction, this algorithm is based on resultant
and completely focused on bivariate polynomials. So we will further
extend it to also tackle polynomials having more than two variables.

As we proceed with our first approach, the algorithm of \cite{Le2001}
first finds candidates for $q$-integer linear types of a given
bivariate polynomial and then obtains the corresponding univariate
polynomials by going through these candidates. The difference is that
it uses resultants to determine candidates and performs bivariate GCD
computations for detecting each candidate.

In order to state its main idea, let $p\in\R[x,y]$ be a polynomial of
positive total degree which is primitive with respect to its either
variable. By Lemma~\ref{LEM:biqinvariant}, an integer pair
$(\lambda,\mu)$ with $\lambda\mu\neq 0$ is a $q$-integer linear type
of $p$ if and only if there exists a factor $f\in\R[x,y]\setminus\R$
of $p$ with the property that $f$ divides $f(q^{\mu}x,q^{-\lambda}y)$
in $\R[x,y]$. Note that such an $f$ must satisfy $\deg_x(f)\deg_y(f)>0$
and $f(x,0)f(0,y)\neq 0$ because $p$ is assumed to be primitive with
respect to its either variable. By a careful study on the structure of
the factor $f$, it is then not hard to see that $f$ divides
$f(q^{\mu}x,q^{-\lambda}y)$ in $\R[x,y]$ if and only if $f$ divides
$f(qx,q^{-\lambda/\mu}y)$ in $\R[x,y]$. Observe that any integer pair
$(\lambda,\mu)$ with $\lambda\mu\neq 0$ is uniquely determined by the
rational $r=-\lambda/\mu$. We have thus shown the following.
\begin{lemma}\label{LEM:res}
  With $p$ given above, a nonzero rational number $r$ gives rise to a
  $q$-integer linear type of $p$ if and only if $\gcd(p,p(qx,q^ry))
  \notin \R$.
\end{lemma}
This implies that for any integer-linear type $(\lambda,\mu)$ of $p$
with $\lambda\mu\neq 0$, the rational number $-\lambda/\mu$ must be a
root of the resultant $\res_y(p,p(qx,q^ry)) \in \R[q^r,x]$ in terms of
$r$, or equivalently, it is eliminated by the content in $\R[q^r]$ of
the resultant with respect to~$x$. Note that such a rational root of a
polynomial in $\R[q^r]$ can be found by matching powers of $q$
appearing in the given polynomial in pairs along with a subsequent
substitution for zero testing. One can find more details in
\citep[\S 5]{Le2001}. Accordingly, we derive a way to produce
candidates for the rationals $-\lambda/\mu$ (and then the $q$-integer
linear types $(\lambda,\mu)$).  After generating candidates, the
algorithm of \cite{Le2001} continues to compute the possible
corresponding univariate polynomial for each candidate
$r=-\lambda/\mu$ by finding a factor $f$ of $p$ that stabilizes
$\gcd(f,f(qx,q^{r}y))$, or more efficiently,
$\gcd(f,f(q^{\mu}x,q^{-\lambda}y))$. This operation actually induces
bivariate polynomial arithmetic over $\R$ and thus may take
considerably more time than Step~8.1 of our algorithm
{\bf MultivariateQILD$_1$}. In order to improve the performance, we
instead proceed by using Step~8 of our algorithm.

We remark that Lemma~\ref{LEM:res} cannot be literally carried over to
polynomials in more than two variables. It is actually not clear how
to directly generalize the algorithm of \cite{Le2001} to the
multivariate case. Nevertheless, using the bivariate-based scheme
indicated by Proposition~\ref{PROP:multi2bi}, this algorithm extends
to the case of polynomials in any number of variables in the same
fashion as our second approach.

The following theorem gives a complexity analysis for the algorithm of
\cite{Le2001} when applied to a polynomial in $\set Z[q,q^{-1}][x,y]$.
\begin{theorem}\label{THM:rescost}
  Let $p\in\set Z[q,q^{-1}][x,y]$. Assume that both the numerator and
  denominator of $p$ have maximum degree $d$ in each variable from
  $\{q,x,y\}$ separately, and let $||p||_\infty=\beta$. Then the
  algorithm of Le takes $\bigO((d^6\log d+d^6\log\beta)\M(d^2)\M(\log
  d+\log\log\beta) \log d\log(\log d+\log\log\beta)+d^6\M(d\log d+d\log
  \beta)\log(d\log d+d\log \beta))$ word operations.
\end{theorem}
\begin{proof}
  With a slight abuse of notation, let $p$ be the input polynomial
  with content with respect to its either variable being removed.
  Then $p\in\set Z[q,x,y]$ and $\log||p||_\infty\in\bigO(d+\log\beta)$.
  The algorithm proceeds to compute the resultant
  $\res_y(p,p(qx,q^ry))$ with $r$ undetermined. By definition, it is
  readily seen that $\res_y(p,p(qx,q^ry))$ is a polynomial in $\set
  Z[q,q^r,x]$ of degree in $q$ at most $3d^2$, degree in $q^r$ at most
  $d^2$ and degree in $x$ at most $2d^2$. Observe that every entry in
  the Sylvester matrix is a monomial in $q^r$.  Thus we have
  $||\res_y(p,p(qx,q^ry))||_\infty\leq||\res_y(p,p(qx,y))||_\infty$,
  which, by Lemma~\ref{LEM:resultant}, is at most
  $B=(2d)!(2d+1)^{2d-1}(d+1)^{2d-1}||p||_\infty^{2d}$. Then $\log B\in
  \bigO(d\log d+d\log\beta)$. Viewing $q^r$ as a new indeterminate $u$
  independent of $q$, we can compute this resultant using a small
  prime modular algorithm, along with an evaluation-interpolation
  scheme: (1) choose $\lceil\log_2(2B+1)\rceil$ primes, each of word
  length $\bigO(\log\log B)$; (2) for every chosen prime $h$, do the
  following: reduce all coefficients of $p$ and $p(qx,uy)$ modulo $h$,
  evaluate both modular images successively at $3d^2$ points for $q$,
  $d^2$ points for $u$ and $2d^2$ points for $x$, compute $6d^6$
  resultants of two polynomials in $\set Z_h[y]$ of degrees in $y$ at
  most $d$, and recover the modular resultant by interpolation; (3)
  reconstruct the desired resultant using the Chinese remainder
  theorem. Neglecting the cost for choosing primes in Step (1), we
  analyze the costs used by Steps~(2)-(3). In Step~(2), the cost per
  prime $h$ for reducing all coefficients modulo $h$ is
  $\bigO(d^2\log\beta\log h)$ word operations. The process of
  evaluation and interpolation is performed in $\bigO(d^5\M(d^2)\log d)$
  arithmetic operations in~$\set Z_h$. Each resultant over $\set Z_h[y]$
  can be computed using $\bigO(\M(d)\log d)$ arithmetic operations 
  in~$\set Z_h$, yielding $\bigO(d^6\M(d)\log d)$ arithmetic operations
  in $\set Z_h$ in total for this step. Notice that the cost for each
  arithmetic operation in $\set Z_h$ is $\bigO(\M(\log h)\log\log h)$
  word operations. Also notice that every chosen prime $h$ is of word
  length $\log h\in\bigO(\log d+\log\log \beta)$.  Thus Step~(2) in
  total takes $\bigO((d^6\log d+d^6\log\beta)\M(d^2)\M(\log d+\log\log\beta)
  \log d\log(\log d+\log\log\beta))$ word operations.  In Step~(3),
  the Chinese remainder theorem requires $\bigO(d^6\M(d\log d+d\log\beta)
  \log(d\log d+d\log \beta))$ word operations. Therefore, computing
  the resultant $\res_y(p,p(qx,q^ry))$ takes
  $\bigO((d^6\log d+d^6\log\beta)\M(d^2)\M(\log d+\log\log\beta)\log d
  \log(\log d+\log\log\beta)
  +d^6\M(d\log d+d\log \beta)\log(d\log d+d\log\beta))$
  word operations. This dominates the costs for subsequent steps
  including finding the rational roots and computing corresponding
  univariate polynomials. The claimed cost follows.
\end{proof}
\begin{corollary}\label{COR:rescost}
  With the assumptions of Theorem~\ref{THM:rescost}, the algorithm of
  Le takes $\softO(d^{10}\log\beta+d^8\log^2\beta)$ word operations
  with classical arithmetic and $\softO(d^8\log\beta)$ with fast
  arithmetic.
\end{corollary}
\subsection{Cost analysis of the factorization-based algorithm}
\label{SUBSEC:factor}
In the current subsection, we introduce another algorithm which is
based on full irreducible factorization of polynomials and works for
polynomials in any number of variables. In order to analyze its cost,
we will briefly describe its main ideas.

The key observation of this algorithm is that, for any $q$-integer
linear polynomial $p\in\R[\bm x]$ of only one type
$(\lambda_1,\dots,\lambda_n)$, the difference of any two vectors from
$\supp(p)$ can be written into the form
$k\cdot(\lambda_1,\dots,\lambda_n)$ for some $k\in\set Z$. This allows
one to readily determine the $q$-integer linearity of any irreducible
polynomial. That is, given an irreducible polynomial $p\in\R[\bm x]$,
take $\bm \alpha\in \supp(p)$ to be such that $\bx{\alpha}$ is the
trailing monomial of $p$ and investigate whether the difference
between $\bm \alpha$ and any other vector from $\supp(p)$ is equal to
a scalar multiple of the same integer vector. One thus immediately
establishes a factorization-based algorithm for computing the
$q$-integer linear decomposition of a polynomial in $\R[\bm x]$: (1)
first perform the full irreducible factorization of the input
polynomial over~$\R$, then (2) determine the $q$-integer linearity of
each irreducible factor and finally (3) regroup all factors of the
same $q$-integer linear type.

A careful study of the above algorithm leads to the following
complexity.
\begin{theorem}\label{THM:faccost}
  Let $p$ be a polynomial in $\set Z[q,q^{-1}][x,y]$. Assume that both
  the numerator and denominator of $p$ have maximum degree $d$ in each
  variable from $\{q,x,y\}$ separately, and let $||p||_\infty=\beta$.
  Then the factorization-based algorithm described above requires
  $\softO(d^9\log^2\beta)$ word operations with classical arithmetic
  and $\softO(d^8\log\beta)$ with fast arithmetic.
\end{theorem}
\begin{proof}
  Computing a complete factorization of $p$ into irreducibles over
  $\set Z[q,q^{-1}]$ dominates the other costs of the algorithm. This
  is essentially the complexity of factoring in $\set Z[q][x,y]$, for
  polynomials bounded by degree $d$ in all variables ($q$, $x$
  and~$y$).  While we do not know of an explicit analysis of this
  complexity (beyond being in polynomial-time, since \citep{Kalt1985}),
  the algorithm of \cite{Gao2003} can be applied and analyzed over the
  function field $\set Q(q)$, and appears to require
  $\softO(d^9\log^2\beta)$ word operations with classical arithmetic
  and $\softO(d^8\log\beta)$ with fast arithmetic.
\end{proof}

\begin{remark}
  Recall from Corollary~\ref{COR:rescost} that the algorithm of Le
  takes $\softO(d^{10}\log\beta+d^8\log^2\beta)$ word operations with
  classical arithmetic and $\softO(d^8\log\beta)$ with fast arithmetic. 
  This compares to the above algorithm based on factorization which 
  requires $\softO(d^9\log^2\beta)$ word operations with classical 
  arithmetic and $\softO(d^8\log\beta)$ with fast arithmetic.  
  All of these compare to Corollary~\ref{COR:multicost1} 
  (or Corollary~\ref{COR:multicost2}) with $n=2$, which reads that 
  our algorithm when restricted to the bivariate case takes 
  $\softO(d^{8}+d^{6}\log^2\beta)$ word operations with classical 
  arithmetic and $\softO(d^{7}+d^{6}\log\beta)$ with fast arithmetic.
\end{remark}

\section{Implementation and timings}\label{SEC:tests}
We have implemented both of our algorithms in {\sc Maple~2018} in the
case where the domain $\R$ is the ring of polynomials over $\set
Z[q,q^{-1}]$.  The code is available by email request. In order to get
an idea about the efficiency of our algorithms, we have compared their
runtimes, as well as the memory requirements, to the performance of
our Maple implementations of the two algorithms discussed in the
preceding section.

The test suite was generated by
\begin{equation}\label{EQ:test}
  p = P_0\prod_{i=1}^m\num(P_i(\bx{\lambda_i})),
\end{equation}
where $n,m\in\set N$,
\begin{itemize}
\item $P_0\in \set Z[q][x_1,\dots,x_n]$ is a random polynomial with
  $\deg_{x_1,\dots,x_n}(P_0)=\deg_q(P_0)=d_0$,
\item the $\bm\lambda_i\in\set Z^n$ are random integer vectors each of
  which has entries of maximum absolute value no more than 10 (note
  that they may not be distinct),
\item $P_i(z)=f_{i1}(z)f_{i2}(z)$ with $f_{ij}(z)\in\set Z[q][z]$ a
  random polynomial of degree $j\cdot d$ for some $d\in\set N$, and
  $\num(\,\cdots)$ denotes the numerator of the argument.
\end{itemize}
Note that, in all tests, the algorithms take the expanded forms of
examples given above as input. All timings are measured in seconds on
a Linux computer with 128GB RAM and fifteen 1.2GHz Dual core
processors. The computations for the experiments did not use any
parallelism.

For a selection of random polynomials of the form \eqref{EQ:test} for
different choices of $n,m,d_0,d$, Table~\ref{TAB:timing} collects the
timings of the algorithm of Le (LQILD), the algorithm based on
factorization (FQILD) and our two algorithms (MQILD$_1$, MQILD$_2$).
The dash in the table indicates that with this choice of
$(m,n,d_0,d)$, the corresponding procedure reached the CPU time limit
(which was set to 12 hours) and yet did not return.

\begin{table}[ht]
  \centering
  \begin{tabular}{l|rrrr}
    $(n,m,d_0,d)$ & LQILD & FQILD & MQILD$_1$ & MQILD$_2$\\\hline
    $(2, 1, 1, 1)$ & 5408.48 & 0.04 & 0.01 & 0.01 \\
    $(2, 1, 5, 1)$ & 8381.99 & 0.06 & 0.03 & 0.03 \\
    $(2, 1, 10, 1)$ & --~~ & 0.19 & 0.04 & 0.04 \\
    $(2, 1, 20, 1)$ & --~~ & 0.63 & 0.09 & 0.09 \\
    $(2, 1, 30, 1)$ & --~~ & 1.47 & 0.13 & 0.10 \\
    $(2, 1, 40, 1)$ & --~~ & 2.55 & 0.24 & 0.21 \\
    $(2, 1, 50, 1)$ & --~~ & 6.64 & 0.42 & 0.39 \\
    $(2, 2, 10, 1)$ & --~~ & 0.92 & 0.10 & 0.08 \\
    $(2, 3, 10, 1)$ & --~~ & 3.29 & 0.31 & 0.26 \\
    $(2, 4, 10, 1)$ & --~~ & 5.74 & 0.67 & 0.54 \\
    $(2, 5, 10, 1)$ & --~~ & 18.83 & 2.01 & 1.54 \\
    $(2, 2, 10, 2)$ & --~~ & 4.55 & 0.27 & 0.20 \\
    $(2, 4, 10, 2)$ & --~~ & 114.82 & 4.98 & 4.53 \\
    $(2, 5, 10, 2)$ & --~~ & 264.02 & 25.63 & 24.29 \\
    $(2, 3, 10, 2)$ & --~~ & 36.14 & 1.38 & 1.21 \\
    $(2, 3, 10, 3)$ & --~~ & 169.13 & 4.28 & 3.80 \\
    $(2, 3, 10, 4)$ & --~~ & 649.03 & 12.15 & 12.86 \\
    $(2, 3, 10, 5)$ & --~~ & 1554.31 & 31.54 & 33.50 \\
    $(2, 2, 5, 1)$ & --~~ & 0.32 & 0.05 & 0.05 \\
    $(3, 2, 5, 1)$ & --~~ & 1.99 & 0.14 & 0.12 \\
    $(4, 2, 5, 1)$ & --~~ & 11.46 & 0.35 & 0.20 \\
    $(5, 2, 5, 1)$ & --~~ & 183.17 & 0.99 & 0.63 \\
    $(6, 2, 5, 1)$ & --~~ & 1141.32 & 2.58 & 0.98 \\
    $(7, 2, 5, 1)$ & --~~ & 11759.89 & 6.07 & 1.74 \\
    $(8, 2, 5, 1)$ & --~~ & 18153.45 & 10.60 & 5.29 \\
    $(9, 2, 5, 1)$ & --~~ & --~~ & 65.53 & 38.12 \\
    $(10, 2, 5, 1)$ & --~~ & --~~ & 176.25 & 89.87 \\
    \hline
  \end{tabular}
  \medskip
  \caption{\small Comparison of all four algorithms for a collection
    of polynomials $p$ of the form \eqref{EQ:test}.}
  \label{TAB:timing}
\end{table}

\section{Conclusion}\label{SEC:conclusion}
In this paper we have presented two new algorithms for computing the
$q$-integer linear decomposition of a multivariate polynomial over any
UFD of characteristic zero. When restricted to the bivariate case,
both algorithms reduce to the same algorithm.  For the sake of
comparison, we included an algorithm based on full irreducible
factorization of polynomials. Compared with the known algorithm of
\cite{Le2001} and this factorization-based algorithm in the bivariate
case, our algorithm is considerably faster.  In practice, both our
algorithms are also more efficient than these two algorithms. In
addition, we have extended and improved the original contribution of
Le and provided complexity analysis for the improved version. We
remark that both our algorithms have much better performances than the
other two algorithms in the case where the coefficient domain contains
algebraic numbers.

\section*{Acknowledgments}
We would like to thank the anonymous referees for many useful and
constructive suggestions. In particular, we would like to express our
gratitude to one referee, whose comments led to the current simple
forms of the proofs of Lemmas~4.1-4.2 and considerably improved the
structure and content of this work.  Most of the work presented in
this paper was carried out while Hui Huang was a Post Doctoral Fellow
at the University of Waterloo.
This research was partly supported by the Natural Sciences and
Engineering Research Council (NSERC) Canada. Hui Huang was also
supported by the Fundamental Research Funds for the Central
Universities.

\bibliographystyle{elsarticle-harv}
\newcommand{\Gathen}{\relax}\newcommand{\Hoeij}{\relax}\newcommand{\Hoeven}{\relax}\def\cprime{$'$}
  \def\cprime{$'$} \def\cprime{$'$} \def\cprime{$'$} \def\cprime{$'$}
  \def\cprime{$'$} \def\cprime{$'$} \def\cprime{$'$} \def\cprime{$'$}
  \def\polhk#1{\setbox0=\hbox{#1}{\ooalign{\hidewidth
  \lower1.5ex\hbox{`}\hidewidth\crcr\unhbox0}}} \def\cprime{$'$}

\end{document}